\documentclass[11pt]{article}  

\usepackage[letterpaper,margin=1in]{geometry}

\usepackage{amsthm}
\usepackage{amsmath, amssymb}
\usepackage{cite}
\usepackage{url}

\usepackage{appendix}

\usepackage{color}
\usepackage{algorithm}
\usepackage[noend]{algpseudocode}
\usepackage{changepage}
\usepackage[framemethod=tikz]{mdframed}
\usepackage{epstopdf}
\usepackage{graphicx}
\usepackage{multirow}

\usepackage{caption}

\usepackage{hyperref}
\hypersetup{
    unicode=false,          % non-Latin characters in Acrobat’s bookmarks
    colorlinks=true,        % false: boxed links; true: colored links
    linkcolor=red,          % color of internal links (change box color with linkbordercolor)
    citecolor=green,        % color of links to bibliography
    filecolor=magenta,      % color of file links
    urlcolor=cyan           % color of external links
}

\usepackage[capitalize]{cleveref}

\algnewcommand\algorithmicswitch{\textbf{switch}}
\algnewcommand\algorithmiccase{\textbf{case}}

\algdef{SE}[SWITCH]{Switch}{EndSwitch}[1]{\algorithmicswitch\ #1\ \algorithmicdo}{\algorithmicend\ \algorithmicswitch}%
\algdef{SE}[CASE]{Case}{EndCase}[1]{\algorithmiccase\ #1}{\algorithmicend\ \algorithmiccase}%
\algtext*{EndSwitch}%
\algtext*{EndCase}%

\algnewcommand\algorithmicwithprob{\textbf{with probability}}
\algnewcommand\algorithmicotherwise{\textbf{otherwise}}

\algdef{SE}[WithProb]{WithProb}{EndWithProb}[1]{\algorithmicwithprob\ #1\ \algorithmicdo}{\algorithmicend\ \algorithmicwithprob}%
\algdef{Ce}[Otherwise]{WithProb}{Otherwise}{EndWithProb}
  {\algorithmicotherwise}%
\algtext*{EndWithProb}%
\algtext*{EndOtherwise}%
\newcommand{\poly}{\operatorname{\text{{\rm poly}}}}
\newcommand{\floor}[1]{\lfloor #1 \rfloor}

\newtheorem{theorem}{Theorem}[section]
\newtheorem{lemma}[theorem]{Lemma}

\newtheorem{corollary}[theorem]{Corollary}
\newtheorem{proposition}[theorem]{Proposition}

\newtheorem{definition}[theorem]{Definition}
\crefname{@theorem}{theorem}{theorems}
\Crefname{@theorem}{Theorem}{Theorems}

\newcommand{\newparagraph}[1]{\smallskip\paragraph{#1}}

\newcommand{\FullOrShort}{full}

\ifthenelse{\equal{\FullOrShort}{full}}{
	
	  \newcommand{\fullOnly}[1]{#1}
	  \newcommand{\shortOnly}[1]{}
		\newcommand{\IncludePictures}[1]{#1}
		\newcommand{\algorithmsize}{\small}

  }{

	  \newcommand{\fullOnly}[1]{}
	  \newcommand{\shortOnly}[1]{#1}
		\newcommand{\IncludePictures}[1]{}
		\newcommand{\algorithmsize}{\small}
  }

\begin{document}

\title{Near Optimal Leader Election in Multi-Hop Radio Networks%
%\iffalse
%\shortOnly{%
%\\[-0.1cm] {\normalsize(Extended Abstract%
%\footnote{Given the short 10-page single-column 11-point font format of this submission, most technical details and proofs 
%are deferred. They can be found in the \emph{full version} that is attached after the \emph{extended abstract} version. We apologize for the long
%full version and hope that the extensive overview given in \Cref{sec:overview} is helpful in understanding and reviewing this work. We also recommend reading the overview section in the full version, as it contains three additional figures. It is worth remarking
%that the 20-page 10-point font double-column format of the SODA conference proceedings will allow us to essentially include the
%full paper with all the proofs and details.\vspace{0.1cm}} )}}\fullOnly{\\[-0.1cm] {\normalsize(Full Version)}}
%\fi
}

\author{
Mohsen Ghaffari \\MIT\\ \texttt{ghaffari@mit.edu}
\and
Bernhard Haeupler\\ Microsoft Research \\ \texttt{haeupler@cs.cmu.edu}
}

\date{}

\maketitle

\begin{abstract}
We present distributed randomized leader election protocols for multi-hop radio networks that elect a leader in almost the same time $T_{BC}$ required for broadcasting a message. For the setting without collision detection, our algorithm runs with high probability\footnote{We use the phrase with high probability (w.h.p.) to indicate a probability at least $1-\frac{1}{n^c}$, for any constant $c\geq 2$. As usual with the randomized algorithms, larger constants $c$ can be obtained at the expense of larger constants in the round complexity.} in $O(D \log \frac{n}{D} + \log^3 n) \cdot \min\{\log\log n,\log \frac{n}{D}\}$ rounds on any $n$-node network with diameter $D$. Since $T_{BC} = \Theta(D \log \frac{n}{D} + \log^2 n)$ is a lower bound, our upper bound is optimal up to a factor of at most $\log \log n$ and the extra $\log n$ factor on the additive term. This algorithm is furthermore the first $O(n)$ time algorithm for this setting.

\smallskip

Our algorithm improves over a 25 year old simulation approach of Bar-Yehuda, Goldreich and Itai with a $O(T_{BC} \log n)$ running time: In 1987 they designed a fast broadcast protocol and subsequently in 1989 they showed how it can be used to simulate one round of a single-hop network that has collision detection in $T_{BC}$ time. The prime application of this simulation was to simulate Willards single-hop leader election protocol, which elects a leader in $O(\log n)$ rounds with high probability and $O(\log \log n)$ rounds in expectation. While it was subsequently shown that Willards bounds are tight, it was unclear whether the simulation approach is optimal. Our results break this barrier and essentially remove the logarithmic slowdown over the broadcast time $T_{BC}$. This is achieved by going away from the simulation approach. 

\smallskip

We also give a distributed randomized leader election algorithm for the setting with collision detection (even with single-bit messages) that with high probability runs in $O(D + \log n \log \log n) \cdot \min\{\log \log n, \log \frac{n}{D}\} = O(D + \log n) \cdot O(\log^2 \log n)$ rounds. This round complexity is optimal up to $O(\log \log n)$ factors and it improves over a deterministic algorithm that requires $\Theta(n)$ rounds independently of the diameter $D$.

\smallskip

Our almost optimal leader election protocols are especially important because countless communication protocols in radio networks use leader election as a crucial first step to solve various, seemingly unrelated, communication primitives such as gathering, multiple unicasts or multiple broadcasts. Even though leader election seems easier than these tasks, its best-known $O(T_{BC} \log n)$ running time had become a bottleneck, preventing optimal algorithms. Breaking the simulation barrier for leader election in this paper has subsequently led to the development of near optimal protocols for these communication primitives. 
\end{abstract}

%\newpage 

%-----------------------------------------------------------------------------------
\section{Introduction}
In this paper we present a randomized distributed algorithm for electing a leader in a radio network without collision detection, which runs in time almost equal to the time required for broadcasting a message. This improves over the 1989 algorithm of Bar-Yehuda, Goldreich and Itai\cite{BII93}.

Leader election, the task of nodes agreeing on the election of a single node in a network, is one of the most fundamental problems in distributed computing. Many tasks in the distributed settings require (or benefit from) having one designated ``organizer" node and leader election is the primitive providing this organizer. Due to its importance, leader election has been studied in many different network settings. 

The setting we are interested in is radio networks. One of the standard models to study these networks is the \emph{radio network model} presented in \cite{CK} (see \Cref{subsec:other-models} for other models). In this model, the communication between the nodes happens in synchronous rounds. In each round, each node either \emph{transmits} a logarithmic size message or remains silent \emph{listening}. If a node is silent and exactly one of its neighbors is transmitting, then it receives a message, namely that of the transmitting neighbor. Other nodes do not receive a message. In particular, if a node has multiple transmitting neighbors, it gets a collision. Depending on the model, such a collision might be detected at this node (the receiver) or it might be considered indistinguishable from the case where no neighbor is transmitting.

This interfering behavior of transmissions makes even basic communication tasks challenging. Since the introduction of the radio network model in 1985, several hundred research papers have given more and more efficient solutions to communication problems such as single-message broadcast, leader election, aggregation, multiple unicasts or broadcasts. The two first and most influential papers in this direction are \cite{BGI1} and \cite{BGI2} published in 1987 and 1989 by Bar-Yehuda, Goldreich and Itai (BGI). In the first paper \cite{BGI1}, BGI presented the Decay protocol as an efficient single-message broadcast protocol for radio networks. Since then, the Decay protocol has been one of the main methods for coping with collisions of radio networks. In the second work \cite{BGI2}, BGI use the Decay protocol to emulate single-hop networks with collision detection in multi-hop networks without collision detection, with a slowdown factor equal to the broadcast time $T_{BC}$. The prime application for this emulation was to transfer results for leader election on single hop networks with collision detection to the multi-hop networks without collision detection. In particular, this allowed for simulating a leader election algorithm of Willard \cite{willard} in multi-hop networks without collision detection. This emulation approach elects a leader in expected time $O(T_{BC} \log \log n)$ rounds and in $O(T_{BC} \log n)$ rounds with high probability. 

The obvious question asked by BGI~\cite{BGI2} was whether this time can be improved. Despite lots of works, this question remained mainly unanswered except knowing the optimal complexity of each of the pieces of the emulation approach: Novel upper and lower bounds showed that in a diameter $D$ network $T_{BC}$ equals $\Theta(D \log \frac{n}{D} + \log^2 n) = O(n)$\cite{ABLP, KM, KP03, CR}. Moreover, \cite{NakOla02} showed that $\Omega(\log n)$ rounds are needed for a high probability leader election in single hop networks with collision detection. Thus, the remaining question now was about whether the whole emulation approach is optimal.

We break this simulation barrier for leader election by presenting an algorithm which parts from the simulation paradigm and achieves time complexity almost $T_{BC}$. More precisely, this algorithm runs with high probability in $O((D \log \frac{n}{D} + \log^3 n) \cdot \min\{\log \log n, \log \frac{n}{D}\})$ rounds --- which is in $O(n)$ rounds --- on any $n$ node network with diameter $D$. This is almost optimal since $\Omega(D \log \frac{n}{D} + \log^2 n)$ is a lower bound. We also give an algorithm for radio networks with collision detection. This algorithm runs in near optimal time $O(D + \log n \log \log n) \cdot \min\{\log \log n, \log \frac{n}{D}\}$ --- which is also in $O(n)$ --- almost matching the respective $\Omega(D + \log n)$ lower bound. We note that these two are the first algorithms that solve the leader election problem in essentially the time needed to broadcast one message (each in the related setting). 

Aside from the complexity of leader election, there is another side to the story which makes leader election important in multi-hop networks: many communication protocols in multi-hop radio networks traditionally use leader election as a crucial first step to solve various, seemingly unrelated, communication tasks. For instance, the fast solutions for many multi-message multi-source broadcast problems first elect a leader, then gather the messages at this leader, and then broadcast them from there\cite{CGL, KK, CKR, GPX05}. This approach helps in managing the contention between the different messages and leads to more efficient algorithms. Even though leader election seems easier than these communication tasks, its $O(T_{BC} \log n)$ best-known bound had become a bottleneck and had kept time complexities of these other problems unresolved as well. Our results solve this issue and set the stage for obtaining near optimal algorithms for many other natural communication primitives that rely on leader election. In particular, subsequent to a preliminary version of this paper, in \cite{SODA-COMM}, we make use of the leader election algorithm presented here while improving the round complexity of these communication primitive to near-optimality.

\section{Related Work}
Leader election is a fundamental problem in distributed computing and it has received vast amount of attention under various communication models and assumptions~\cite{Lynch96}. This problem becomes considerably more challenging in wireless networks. 

\subsection{Theoretical Models of Wireless Networks}\label{subsec:other-models}

As is usual with most practical settings, there is a large variety of theoretical models for wireless networks. The standard LOCAL model of distributed computing, in which in each time unit each node can send one message to all of its neighbors in a graph, can be seen as a very crude first-order approximation of wireless networks. While this model captures the ``local broadcast" nature of wireless networks it completely ignores  physical layers issues such as interference. As such, it can be construed to be abstracting the scenario in which one is working above the MAC layer. In contrast to this simple and clean model, there are many more that go below the MAC layer and include the physical layer issues. 

One classical model that considers both the ``local broadcast" nature and ``interference" is the graph-based \emph{radio network model}. The single-hop version of this model was considered in the 70's---see \cite{Hayes78, TsyMik78, Capetanakis}---and its multi-hop version was introduced by Chlamatac and Kutten\cite{CK} in 1985. This radio network model has been studied since then and is still a topic of ongoing investigation (see, e.g., these recent papers \cite{NakOla02,CMS03,GLS11,aghk2014,CK10,AABCHK13,BCC12,KP09,chlebus2012electing,SODA-COMM,KK,CGL,CKR,GPX05}). Since many basic problems are challenging in this model when studied for general graphs, there has been also interest in restricted families of graphs that assumed to be capturing most of the practical scenarios. A prominent example is that of restriction to Unit Disc Graphs\cite{clark1991unit}. 

There are also other models of wireless networks which include further details about the physical layer. A wide range of these can be put under the general title of {\em SINR-style} models; i.e., models in which communication behavior is determined by the ratio of signal to noise and interference (see e.g.~\cite{moscibroda:2006,scheideler:2008,kesselheim:2011,halldorsson:2012b,jurdzinski:2013, Daum:2013}), and where a key characteristic is that the signal power decays (polynomially) with the traveled distance. The models in this category differ in many aspect such as: assumptions about the geographical spread of the nodes in the area and whether the are is obstacle free and uniform or not, signal decay patters, whether nodes have fixed powers or controllable powers, exact knowledge of network parameters, etc. 
One can move even further and include important issues such as multi-path effects, non-uniform signal propagation, analog representation of information where packets are viewed as waveforms and one can apply powerful approaches such interference alignment \cite{Cadambe:2008} and analog network coding\cite{katti2007embracing}, multi-channel models \cite{kyasanur2005capacity}, Multi-Input Multi-Output (MIMO) extensions\cite{gesbert2003theory}. See \cite{halldorsson2014modeling} for a recent discussion about some of these issues.

As the above discussions indicate, the range of possibilities for modeling physical issues in wireless networks is quite wide and one can always include further details. However, typically, as more and more details are included, the algorithmic solutions become increasingly more involved and dependent on the details of the model. Therefore while simpler models might lead to unrealistically simple solutions more specific model run in danger to supply less high-level insights in general and also produce results that are too dependent on the specific model assumptions to be transferable. %Because of this, less-detailed models are still being studied actively. In particular, in the community of distributed algorithms for wireless networks, still a majority of the attention is divided between the graph-based model and the simplest versions of the SINR model. 

In this paper, we focus on the graph based model and make a progress on one of its fundamental open questions, namely the leader election problem. In the next subsection, we review the related work about leader election in graph based models. 

\subsection{Leader Election in Graph Based Models}
Leader election in the graph based radio network model has received a lot of attention (see e.g. \cite{Hayes78, TsyMik78, Capetanakis, GW85, willard, NakOla02, GLS11, CMS03, KM, BGI2, KP09}). In the following, we review this line of work in two categories of single-hop networks and multi-hop networks. 

\medskip
\noindent\textbf{Single-Hop Radio Networks:} The study of leader election in radio networks started with the special case of single-hop networks, where the network is a complete graph. The story goes back to 70's and 80's, when~\cite{Hayes78, TsyMik78, Capetanakis} independently showed that in the model with collision detection, the problem can be solved in $O(\log n)$ rounds deterministically, and this was shown to be optimal for deterministic algorithms by $\Omega(\log n)$ lower bound of~\cite{GW85}. On the randomized side of the problem in the model with collision detection, even though the expected time was improved to $O(\log \log n)$~\cite{willard, NakOla02}, the high probability time remained $O(\log n)$ in both. These bounds were proven to be tight by $\Omega(\log \log n)$ lower bound on the expected time of uniform protocols~\cite{willard}, and $\Omega(\log n)$ lower bound for the high probability time of uniform protocols~\cite{NakOla02}. The assumption of uniformity in the latter result was later removed \cite{GLS11}.

In the single-hop networks without collision detection, for deterministic algorithms~\cite{CMS03} presented matching upper and lower bounds of $\Theta(n \log n)$. For randomized bounds, \cite{KM} showed that $\Omega(\log n)$ is a lower bound on the expected time, and \cite{JS05} showed that $\Theta(\log^2 n)$ is the tight bound for the high probability time by presenting $O(\log^2 n)$ upper bound and $\Omega(\log^2 n)$ lower bound. 

These bounds, altogether, in principle settle the time complexity of the single-hop case.

%\begin{itemize}
%\item Metcalfe and Boggs, Ethernet: Distributed Packet Switching for local computer networks, $O(1)$ expected and $O(\log n)$ whp, knowing $n$ and presented for ethernet 1976
%\item Hayes, An adaptive technique for local distribution. 1978 deterministic $O(\log n)$ With-CD
%\item Tsybakov, B.S., Mikhailov, V.A.: Free synchronous packet access in a broadcast channel with feedback., 1978 with-CD $O(\log n)$
%\item Capetanakis, Tree algorithms for packet broadcast channels, 1979 deterministic $O(\log n)$ With-CD
%\item Greenberg, Winograd, a lower bound on the time needed in the worst case to resolve conflicts deterministically in multiple access channels, 1985, deterministic $\Omega(\log n)$ With-CD
%
%\item Log-Logarithmic selection resolution protocols, Willard, 1986, randomized $O(\log \log n)$ expected and $O(\log n)$ whp With-CD
%\item Kushilevitz Mansour, $\Omega(log n)$ for randomized whp with No-CD 1989
%
%\item Nakano, Olariu, Uniform leader election protocols for radio networks, 2002, good survey and  $\log \log n+O(\log \frac{1}{\eps})$ whp With-CD
%\item Clementi, Monti, Silvestri, distributed Broadcast in radio networks with unkown topology, $\Theta(n \log n)$ upper and lower bounds for deterministic with No-CD 
%\item Jurdzinski, Stachowiak, Probabilistic Algorithms for the Wakeup Problem in Single-Hop Radio Networks, 
%\end{itemize}

%\bigskip
\medskip
\noindent\textbf{Multi-Hop Radio Networks:} In contrast to the single-hop special case, the complexity of the general case of multi-hop networks did not see much progress, after the initial results. 

The research about theoretical problems in multi-hop radio networks essentially started with the pioneering papers of Bar-Yehuda, Goldreich and Itai (BGI) ~\cite{BGI1, BGI2}. In the first paper, BGI devised the \emph{Decay} protocol as a solution for single-message broadcast problem, resulting in almost optimal broadcast time of $O(D \log n + \log^2 n)$. This protocol later became the standard approach in coping with collisions of radio networks (see e.g.~\cite{BGI2, BII93, CMS01, GPX05}). Provided by this almost optimal broadcast algorithm, and given that the case of leader election in single-hop radio networks was well-studied, a natural idea was to simulate the `single-hop' leader election algorithms over multi-hop radio networks. Along this idea, in the second paper, BGI used Decay protocol to emulate a single-hop radio network with collision detection on top a multi-hop radio network without collision detection. As the prime application, they used this emulation to simulate Willard's single-hop leader election algorithm~\cite{willard} in multi-hop radio networks without collision detection. This resulted in time complexity of $O((D\log n + \log^2 n) \log n) = O(n \log^2 n)$ for a with high probability result (and also $O((D\log n + \log^2 n) \log \log n) = O(n \log n \cdot \log \log n)$ for expected time). %BGI~\cite{} asked ``\emph{it is interesting to investigate whether [this] algorithm for leader election in arbitrary radio networks can be improved}''.

Given this efficient algorithm, the remaining question was how to improve it to optimality. One idea would be to use a better leader election algorithm of single-hop networks, but given lower bounds $\Omega(\log \log n)$ on the expected time~\cite{willard} and $\Omega(\log n)$ for high probability results~\cite{NakOla02}, there was no hope in that direction.
%
%\bigskip
%
%{\bfseries Some Reformulation is needed from here on.}
%\bigskip

The next idea was to improve upon the Decay broadcast algorithm. By modifying the Decay protocol, Czumaj and Rytter~\cite{CR} and Kowalski and Pelc\cite{KP03} reduced the time complexity of broadcast from $O(D \log n + \log^2 n)$ to $T_{BC}=O(D\log \frac{n}{D}+\log^2 n)$, which is known to be optimal in the light of $\Omega(D \log{\frac{n}{D}})$ lower bound of~\cite{KM} and $\Omega(\log^2 n )$ lower bound of \cite{ABLP}. Albeit not being published explicitly, by providing a substitute for the old Decay in BGI's framework, this new Decay changed the time complexity of leader election (using simulation approach) to $O(T_{BC} \log n) = O(n \log n)$ for a with high probability algorithm (and $O(T_{BC} \log \log n) = O(n \log \log n)$ expected time). Given that now both elements of the emulation --- single-hop leader election algorithm and broadcast algorithm --- were optimal, the remaining interesting question was ``\emph{can one improve upon the leader election time bound by going away from the simulation approach?}''. In this paper we answer this question in affirmative.

For networks with collision detection, Kowalski and Pelc \cite{KP09} presented an $O(n)$ deterministic algorithm. This highlighted the difference between models with and without collision detection as a lower bound of $\Omega(n \log n)$ was known for deterministic leader election without collision detection even for single-hop networks~\cite{CMS03}. We remark here that the time complexity of this algorithm remains $\Theta(n)$ even when diameter $D$ of the network is small.

We also note that, simultaneous with the preliminary version of this paper, Chlebus, Kowalski and Pelc\cite{chlebus2012electing} presented a randomized algorithm for the setting without collision detection that runs in $O(n)$ rounds in expectation and in $O(n\log n)$ rounds with high probability. These running times remain $\Omega(n)$ regardless of how small the diameter of the graph is. However, we remark that the result of \cite{chlebus2012electing} has the nice property that the safety guarantee of the leader election is deterministic, and only the running time guarantees are probabilistic. In contrast, in our algorithms, the guarantee is that with high probability one leader is elected. The same paper\cite{chlebus2012electing} also provides a deterministic solution with round complexity $O(n \log^{1.5}\sqrt{\log \log n})$.

%In this paper, we answer to the above open question by presenting near optimal an algorithm which solves the problem in almost $T_{BC}$ time. More precisely, this algorithm solves leader election the model without collision detection in time $O((D \log \frac{n}{D} + \log^3 n) \cdot \min\{\log\log n,\log \frac{n}{D}\})$ with high probability. Moreover, of the model with collision detection (in fact, for the significantly weaker beeping model), we present an algorithm with time complexity $O(D + \log n \log \log n) \cdot \min\{\log \log n, \log \frac{n}{D}\}$.

%\medskip
%\noindent\textbf{Applications and Relations with Other Problems:} Typically, leader election is the first and most basic step for algorithms where some sort of a centralized control is desired. Being so, leader election algorithms have been used as one of main building blocks, in an extensive number of works of radio networks field (see e.g.~\cite{BII93, CKR, KK}). Moreover, surprisingly this seemingly simple task has became the bottleneck in many of these works(see e.g. ~\cite{BII93, KK}). By improving the time complexity of leader election to optimality, our results open up the road for progress in other important primitives of the radio networks. Such progress was achieved for a number of basic primitives such as aggregation, $k$-message broadcast and $k$ point-to-point communications~\cite{SODA-COMM}.

\section{Preliminaries}

\subsection{Network Models}\label{sec:model}

We consider the standard multi-hop radio network model\cite{CK, BGI1, BGI2}. In this model the network is represented by a connected undirected graph $G = (V, E)$ with $n=|V|$ nodes and diameter $D$. Communication in such a network takes place in synchronous rounds; in each round, each node is either listening or transmitting a $\Theta(\log n)$-bit packet. In each round, each node $v \in V$ can receive a packet only from its neighbors and only if $v$ itself is not transmitting in that round. If two or more neighbors of $v$ transmit in a round, then these transmissions \emph{collide} at $v$ and $v$ does not receive any packet. In this case, we consider two model variants: (1) the model with no \emph{collision detection} (CD) where the node $v$ can not distinguish this collision from silence, and (2) the model with CD where $v$ gets to know that a collision happened. Another small distinction is whether a synchronous wake-up is assumed in which all nodes start participating at time zero or whether only some (arbitrary subset of) nodes are awake initially. In the later case nodes that are not awake are always listening until they are woken up by the first received message or detected collision (see also \Cref{sec:LEproblem}).

Instead of studying the radio network model with collision detection directly, we choose a strictly weaker model, namely the \emph{beep model} as given in~\cite{CK10,AABCHK13}. A beep network works in synchronous rounds. In each round, each node can either \emph{beep} (transmit) or remain silent. At the end of a round, each silent node gets to know whether at least one of its neighbors was beeping or not. We note that the beep model can be seen as a radio network model with collision detection and $1$-bit packets but with the additional weakening limitations that nodes can not distinguish between one neighbor sending a $1$ or $0$ or between the cases where exactly one or more than one neighbor is beeping. This extremely basic communication model is an interesting weakening of the standard model with collision detection, from both theoretical and practical viewpoints. Any algorithm designed for the beeping model can be directly used for the standard model with collision detection. However, designing such algorithms is typically more challenging. On the practical side, it has been argued that the beeping model can be implemented easily in most environments for example using carrier sensing or simple radios~\cite{CK10}.

All algorithms we study in this paper are randomized and distributed. Furthermore, all stated running times hold with high probability, that is with probability $1 - n^{-\Theta(1)}$ (in contrast to merely in expectation). As is standard for distributed algorithms, we assume that nodes have no knowledge about the topology except for knowing $n$ and $D$ (up to constant factors). We remark that the assumption of knowing $D$ can be removed easily without any asymptotic loss in time bounds, using standard double-and-test parameter estimation techniques. We also assume that nodes create an identifier by independently sampling $\Theta(\log n)$ bits. With high probability these identifiers will be unique IDs. 

%an upper bound $N$ on the number of nodes $n$ with $N = n^{O(1)}$. From this, given the algorithms presented here, one can obtain a constant factor estimation of $D$ and $k$ using standard double-and-test parameter estimation without more than a $\log \log n$-factor loss in efficiency. We defer the details to the full version and assume throughout this paper that $D$ and $k$ are known to the nodes. For simplicity we also assume that the number of messages used is at most polynomial in $n$. We note that this is without loss of generality since for $k > n$ one can bundle messages starting in one node to $O(n)$ larger messages and transmit these with the algorithm presented here using meta-rounds consisting of the same transmissions repeated $k/n$ times to transmit the larger messages. Again, since all $k$-message are at least linear in $k$ this does not change the efficiency by more than a constant factor.

\subsection{Leader Election}

\subsubsection{The Leader Election Problem}\label{sec:LEproblem}

The problem studied in this paper is the Leader Election problem. In this problem the goal is to elect a single node in the network. More formally, we say a (randomized) algorithm $\mathcal{A}$ solves the leader election problem in time $T(n,D)$ if the following holds: Consider an arbitrary undirected network $G=(V, E)$ with $n$ nodes and diameter $D$. Each node randomly chooses a $\Theta(\log n)$ bit ID (which will be unique with high probability). At round $0$ some nodes (at least one) wake up and start running the leader election algorithm. We require that, by round $T(n,D)$, with probability at least $1 - n^{-\Theta(1)}$, each node outputs exactly one ID of an \emph{elected} node in the network and all nodes output the same ID.

We remark that it is possible to wake-up all nodes in a network in broadcast time, that is, in $T_{BC} = \Theta(D \log \frac{n}{D} + \log^2 n)$ rounds without collision detection and $O(D)$ rounds in the beep model. In the same time one can also establish a global clock by informing all nodes about the time zero in which the first nodes woke up. This is done by keeping a counter in the wake-up messages for the setting with collision detection\footnote{For beep networks one can use two beep waves of different speed, one going one hop per round and one going one hop every two rounds. Taking the difference between the two arrival times indicates for every node how long it took the first wave to wake it up and therefore by how much to turn back the clock to be in synch with the nodes woken up in round zero.}% The details are not very interesting and omitted here}
. After this wake-up and synchronization phase, one can solve the leader election problem assuming a synchronous global wake up. Since in the running time analysis of our algorithms an additive term of $T_{BC}$, the time to perform a broadcast, is negligible we may and do restrict ourselves to designing algorithms for the synchronous wake up case throughout the rest of this paper.

\subsubsection{Leader Election vs. Broadcast}\label{sec:BCvsLE}

Next we explain the relationship between the broadcast problem and the leader election problem.

 We first note that the way leader election is solved in the classical simulation approach is to perform several, $\Theta(\log n)$ to be exact, broadcasts. While the approach of this paper is more refined it still heavily relies on broadcast as a fast subroutine. In particular, our running time can be roughly interpreted as being equivalent to $\Theta(\log \log n)$ broadcasts. All of this suggests that leader election is at least as hard of a problem as broadcast. This intuition is further backed up by understanding the leader election problem as a symmetry breaking routine (electing exactly one node) which in the end necessarily needs to broadcast the identity of the elected node. 

In the next lemma we give a simple reduction, which is to our knowledge new, that formalizes this intuition and shows that leader election is at least as hard as a broadcast (up to constant factors). 

\begin{lemma}\label{lem:BCLEreduction}
If there is an algorithm that solves the leader election problem with high probability in $T(D,n)$ rounds for any $n$ node graph of diameter at most $D$, then there exists an algorithm that solves the broadcast problem in $2 T(2D, 2n-1)$ rounds, with high probability.
\end{lemma}
\begin{proof}
Let $\mathcal{A}$ be the algorithm solving the leader election problem and let $G$ be the network on which we want to solve the broadcast problem. In particular, let $v \in G$ be the node starting with the message $M$ to be broadcast. In order to perform the broadcast the nodes will collectively simulate an execution of the leader election algorithm $\mathcal{A}$ on a network $G'$. This network $G'$ consists of two disjoint copies of $G$ except that the two copies of the node $v$ are identified to be one node. The graph $G'$ has therefore $2n-1$ nodes and diameter at most $2D$ and a running the leader election algorithm $\mathcal{A}$ on $G'$ requires at most $T(2D,2n-1)$ rounds. In our simulation every node (except $v$) in $G$ simulates two nodes in $G'$, namely the two copies corresponding to it. Each round of the leader election algorithm in $G'$ is simulated by two rounds in $G$, each simulating the actions of all nodes in one copy of $G$ in $G'$. In particular, each node in $G$ keeps the state of both of its copies in $G'$ and assigns one to the odd and one to the even rounds. Communication and computation are then as governed by the leader election algorithm except for two small changes. Firstly, the node $v$, in order to stay true to the simulation in $G'$, will consider two received messages as a collision if they were received in the two rounds simulating the same round in $G'$. Secondly, in order to spread the broadcast message $M$ the node $v$ and any node that has already received $M$ will piggy-back $M$ to any transmission performed by it. Since $M$ is by assumption small enough to fit into one transmission the increase of the required message size is at most a factor of two and therefore negligible. What remains to be proven is that this way every node receives the broadcast message with high probability during the $2 T(2D,2n-1)$ rounds required for the leader election simulation. To see this we note that $v$ is a cut vertex in $G'$ which guarantees that any communication from one side to the other has to go through $v$. In particular, the side without a leader must in the end be informed about the existence and ID of the leader on the other side. This requires for each node $u$ a sequence of successful transmissions from $v$ to $u$. Since the broadcast message $M$ will be piggy-backed on these transmissions, every node $u$ in $G$ will received at least one transmission containing the message $M$, which is what is needed for a broadcast. 
\end{proof}

\subsubsection{Lower Bounds}\label{sec:LElowerbounds}

%
%It is clear that at time $T(n,D)$ every node must be awake with high probability. If exactly one node $s$ in the network is woken up (by means other than receiving a message from neighbors), then by time $T(n,D)$, with high probability, we have the following: for every node $u$, there is a chain of node wake-ups $s=v_0, v_1, v_2, \dots, v_\ell=u$ where $v_{i+1}$ wakes up because it received a message from $v_{i}$. Hence, in this case, the problem takes at least as long as the time required for performing a broadcast when only the source is awake at the start. This implies that the lower bounds proven in \cite{KM, ABLP}, which show that a single-message broadcast without collision detection (when only the source is awake initially) requires at least $T_{BC} = \Theta(D \log \frac{n}{D} + \log^2 n)$ rounds, also apply to the leader election problem:
%\end{proof}

Next, we explain how known lower bounds for the broadcast problem translate to the leader election problem using the reduction from \Cref{lem:BCLEreduction}.

We first consider the setting without collision detection. The only non-trivial lower bound known for the synchronous wake-up setting without collision detection is the $\Omega(\log^2 n)$ lower bound of \cite{ABLP}. Together with the trivial $\Omega(D)$ lower bound this shows that leader election requires $\Omega(D + \log^2 n)$ rounds. In the non-synchronous wake-up model assumed in this paper however there is a $\Theta(D \log \frac{n}{D})$ broadcasting lower bound of \cite{KM93}. This lower bound applies if in round zero only some nodes (including the node with the broadcast message) wake-up while others need to be woken up by a received transmission. Together with \Cref{lem:BCLEreduction} this gives the following lower bound:

\begin{corollary}\label{lem:lowerboundLEnoCD}
Any (randomized) algorithm in the non-synchronous wake-up model requires at least $\Omega(T_{BC})= \Omega(D \log \frac{n}{D} + \log^2 n)$ rounds to solve the leader election problem in radio networks without collision detection, with high probability. In the synchronous wake-up model at least $\Omega(D + \log^2 n)$ rounds are required. 
\end{corollary}

For the setting with collision detection the $\Omega(\log n)$ lower bound of~\cite{NakOla02,GLS11} for the single-hop model applies and together with the trivial $\Omega(D)$ diameter lower bound gives:

\begin{corollary}\label{lem:lowerboundLEwithCD}
Any (randomized) algorithm requires at least $\Omega(D + \log n)$ rounds to solve the leader election problem in radio networks with collision detection or beep networks, with high probability.
\end{corollary}
 
%On the other hand, for upper bounds, one can use a $T_{BC}$ round broadcast in the case without collision detection to wake up all nodes in the network. Furthermore, the message used for this \emph{wake up  broadcast} synchronizes every node with the source and it also specifies a future time, within $O(T_{BC})$ rounds, at which point every node is guaranteed to be awake (with high probability). After that, one can solve the leader election problem assuming a synchronous global wake up. Since this broadcast-time is negligible in the running time analysis, this allows us to simply design algorithms for the synchronous wake up case which is what we do for the rest of this paper. \textcolor{red}{For the case with collision detection, this global wake up broadcast can be substituted with a $\Theta(D)$ rounds beep wave. \dots}

\subsection{Message Dissemination in Radio Networks: The Decay-Protocols}\label{sec:decay}
In this section we present a recap on the Decay broadcast algorithm from \cite{BGI1} and the modified variant of it presented in \cite{CR}. The Decay algorithm has become one of the standard techniques for resolving collisions in the graph-based radio networks, and as many other papers in this area, our algorithms also use it frequently. 

\subsubsection{The Basic Decay Algorithm}
\paragraph{A phase of Decay:} The basic version of the Decay protocol \cite{BGI1} works as follows: The execution of the algorithm is organized in phases, where each phase consists of $\log n$ rounds. In the $i^{th}$ round of each phase, each node that has a message to send transmits this message with probability $2^{-i}$ and remains silent otherwise. The key property of this protocol is captured by the following simple lemma:

\begin{lemma}[\cite{BGI1}] Consider a node $v$, a phase of the Decay protocol, and suppose that in this phase, at least one neighbor of $v$ has a message to send. Then in this phase, with probability at least $\frac{1}{10}$, $v$ receives at least one message.
\end{lemma}
Now, to see the power of this property, we explain how Bar-Yehuda et al\cite{BGI1} use the Decay protocol for a single-message broadcast: Suppose that the source starts sending a message at time $0$ and for each node $v$, if $v$ receives the message in a phase, it sends it in all the future phases. Considering a path of length $P$, in expectation, it takes the message a constant number of phases to make a progress of one hop. Thus, in expectation, $D$ hops take $O(D)$ phases. A simple application of the Chernoff bound then provides a probability concentration proving that in $O(D+\log n)$ phases, or equivalently in $O(D\log n+\log^2 n)$ rounds, with high probability all the nodes have received the message.

\paragraph{A long-phase of Decay:} In some of our applications, it is convenient to group $\Theta(\log n)$ phases of the Decay protocol together and use the resulting stronger probability guarantee. We define a \emph{long-phase} of Decay to consist of $\Theta(\log n)$ phases of Decay. In each of these $\Theta(\log n)$ phases, each node that started the long-phase with a message to send tries transmitting its message according to the probabilities of one phase of the Decay protocol explained above. That is, in $i^{th}$ round of each phase, each node that started the long phase with a message to send transmits this message with probability $2^{-i}$ and remains silent otherwise. The only difference to performing $\log n$ short phases is that nodes which get newly informed during a long phase do not forward this message until the long-phase is over.

\begin{lemma}[\cite{BGI1}]\label{lem:long-phase} In each long-phase of the Decay protocol, with high probability, the following holds: each node $v$ that has at least one neighbor that is sending a message receives at least one message.
\end{lemma}
It is easy to see that one can use $D$ long-phases of the Decay protocol to perform a single-message broadcast. The result would be an algorithm with round complexity $O(D\log^2 n)$ which is slower than the basic version, but has the appealing property that the growth of the area of nodes that have received the message is very controlled: with high probability, it grows by exactly one hop in each long-phase.

%The Decay algorithm is used to spread information present in some nodes to neighboring nodes or all nodes in a radio network without collision detection. 
\subsubsection{The Fast Decay Algorithm}\label{subsub:fastDecay}
Finally, in some of our applications, we use the faster variant of the Decay protocol presented by Czumaj and Rytter\cite{CR} that performs a single message broadcast in $T_{BC}=O(D\log{\frac{n}{D}}+\log^2 n)$ rounds (instead of the $O(D\log n+\log^2 n)$ rounds achieved by the basic Decay). We refer to this version as the Fast-Decay algorithm. Czumaj and Rytter \cite[Lemma 7.7]{CR} define a sequence which determines the probability of transmission for each round, and for each node that has a message to transmit. Each node that receives a message starts trying for transmitting it from the next round. This sequence gives the property that all nodes receive the message by round $O(D\log \frac{n}{D}+\log^2 n)$, with high probability. 

In this protocol, it is possible that a node at distance $d$ from the source receives the message in round $d$. In some of the applications, we want to have a closer control on the spreading of the message and thus, we add a simple change: for a delay parameter $\delta\geq \log \frac{n}{D}$, a node starts forwarding the message $\delta$ rounds after the first time that it receives the message. We refer to this as Fast-Decay($\delta$). This additional delay gives us the guarantee that a node at distance $d$ receives the message sometime after round $(d-1)\delta$. On the other hand, it delays the progress of the message by $\delta$ rounds per hop and thus at most $D\delta$ rounds in total.

\begin{lemma}[\cite{CR}] \label{lem:fastdecay}
For $\delta = \log \frac{n}{D}$ and a large enough constant $\alpha$, if the Fast-Decay($\delta$) protocol is run for $T = \alpha (D \delta + \log^2 n)$ rounds, then any node $v$ with distance $d$ to the closest node that is initially a sender will with high probability receive a message for the first time between round $(d-1) \delta$ and round $T$. 
\end{lemma}

Lastly, because the sequence of probabilities given in \cite{CR} contains the Decay sequence as a subsequence it is easy to see that the equivalent of \Cref{lem:long-phase} also holds for the Fast-Decay($\delta$) protocol:

\begin{lemma}[\cite{CR}]\label{lem:fastdecay-slowstep} After running the Fast-Decay($\delta$) protocol for $\alpha\log^2 n$ rounds, every node $v$ at distance one hop from the source (or sources) receives a message with high probability.
\end{lemma}

%\begin{lemma}[\cite{CR}] \label{lem:globaldecay}
%For $\delta \geq \log n$, if the Decay($\delta$) protocol is run for $T = 30 (D \delta + \log^2 n)$ rounds, then any node $v$ with distance $d$ to the closest node that is initially a sender will with high probability receive a message for the first time between round $d \delta$ and round $\Theta(d \delta + \log^2 n)$.
%\end{lemma}

\section{Our Results and Their Overview}\label{sec:overview}
In this paper, we show the following two results:

\begin{theorem}\label{lem:mainnoCD}
In radio networks without collision detection, there is a distributed randomized algorithm that in any network with $n$ nodes and diameter $D$ with high probability solves the leader election problem in time
$$T_{noCD} = O\left(D \log \frac{n}{D} + \log^3 n\right) \cdot \min\left\{\log \log n, \log \frac{n}{D}\right\}.$$
 %In particular, we remark that 
%$T_{noCD} = \tilde{O}(T_{BC} + \log^3 n)$ and $T_{noCD} = O(n)$.
\end{theorem} 
%\smallskip

\begin{theorem}\label{lem:mainCD}
In beeping networks (or radio networks with collision detection), there is a distributed randomized algorithm that in any network with $n$ nodes and diameter $D$ with high probability solves the leader election problem in time $$T_{beep} = O\left(D + \log n \log \log n \right) \cdot \min\left\{\log \log n, \log \frac{n}{D}\right\}.$$
%We remark that $T_{beep} = \tilde{O}(D + \log n)$ and $T_{beep} = O(n)$, which are both optimal.
\end{theorem}

\begin{corollary} In both radio network without collision detection and beeping networks, there are randomized leader election algorithms that in any network with $n$ nodes with high probability solve the leader election problem in $O(n)$ rounds.
\end{corollary}

As mentioned in the introduction, the key novelty in these results is going away from the approach of simulating a single-hop network. Instead, we start with a small number of candidates and gradually reduce this number until we have a leader. We hope that the general outline of our approach, which we sketch in the next subsection, can be useful in other wireless settings as well. 

In the rest of this section, we present an overview of these algorithms. All our algorithms try to implement an ideal leader election template which we describe first. The methods for implementing this template differ depending on whether one is in a setting without or with collision detection. We explain the key ideas of these implementations in \Cref{sec:overviewnoCD} and \Cref{sec:overviewCD}, respectively.

\subsection{The Leader Election Template}\label{sec:template}
The main outline of our algorithms and the topmost level of their ideas are as follows.

\newparagraph{Main Outline and the Debates:} Given the number of nodes $n$, we first use sampling to reduce the number of possible candidates for leadership. Each node decides to be a candidate, independently, with probability $\frac{10 \log n}{n}$. A Chernoff bound then shows that with high probability, this leads to at least one and at most $20 \log n$ candidates. To elect a leader among these candidates, we then run in phases, called \emph{``debates''}. In each \emph{debate}, we eliminate at least a constant fraction of the candidates, while keeping the guarantee that always at least one remains. After $O(\log \log n)$ debates, only exactly one candidate remains. At the end, this candidate declares itself as the leader by broadcasting its ID. This outline is presented in Algorithm \ref{alg:LE}.

\begin{algorithm}[H]
\caption{Leader Election Algorithm @ node $u$}
\begin{algorithmic}[1]
\algorithmsize
\WithProb{$\frac{10\log n}{n}$} 
\State $candidate \gets true$ \Comment{$candidate$ is a Boolean variable which indicates whether $u$ is a candidate}
\Otherwise 
\State $candidate \gets false$
\EndWithProb

\Statex
%\State pick $x_u \in_{\mathcal{U}} [0,1]$ uniformly at random
%\If {$x_u < \frac{\log n}{n}$}
%	\State $candidate_u \gets true$
%\Else
%	\State $candidate_u \gets false$
%\EndIf
\For {$i= 1$ to  $\Theta(\log \log n)$}
	\State \textbf{Debate} \Comment{for some candidates, the Boolean variable $candidate$ becomes $false$}
\EndFor
\smallskip
\Statex
\If {$candidate$}
	\State Broadcast $ID_u$
\EndIf
%\smallskip
\State output the received $ID$ as the leader
\end{algorithmic}
\label{alg:LE}
\end{algorithm}

\paragraph{Clusters, Overlay Graph and Communication Actions:}
To achieve the above goal for debates, we need to provide some way of communication between the candidates. For this, in each debate, we grow \emph{clusters} around each candidate, for example, by assigning each non-candidate node to the candidate closest to it. This clustering induces an overlay graph $H$ on the candidates by declaring two candidates to be \emph{adjacent} in the overlay graph iff their clusters are close (this will be made precise later). A pictorial example is shown in \Cref{fig:clusters}.
\IncludePictures{
\begin{figure}[t]
\centering
%%\begin{tabular}{cc}
%\begin{minipage}{3in}
%\includegraphics[width=2.8in]{huge-graph-clustered.pdf}
%%\end{minipage}&
%\end{minipage}
%\begin{minipage}{3in}
%\includegraphics[width=2.6in]{clustergraph-win.pdf}
%\end{minipage}
%%\end{tabular}
\includegraphics[width=0.7\textwidth]{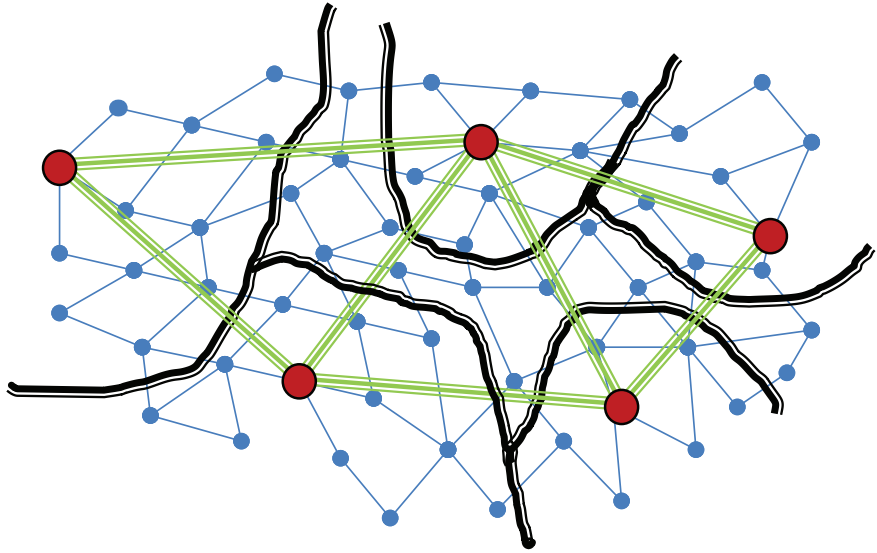}
{\caption{\small The red nodes show the candidates and the area close to them inside the black curves are their respective clusters. The green wide links show the edges of the overlay graph $H$ between the candidates.
%The left figure shows a clustered graph with solid candidate nodes. The resulting overlay graph $H$ on the candidate nodes is depicted in the right figure. The big candidates are the ones that remain after the elimination algorithm is run on $H$.
}}
\label{fig:clusters}
\vspace{-8pt}
\end{figure}
}

This overlay graph $H$ also captures which candidates can communicate with each other using specially designed cluster communication actions. In particular, we design three \emph{communication actions}: an \emph{Uplink} protocol that allows a candidate to send a message to the nodes in its cluster, an \emph{Intercommunication} protocol that allows adjacent clusters to exchange information, and a \emph{Downlink} protocol that allows nodes in a cluster to send a message to their candidate. A diagram depicting these actions is presented in \Cref{fig:Communications}.

\IncludePictures{
\begin{figure}[t]
\centering
\includegraphics[width=0.8\textwidth]{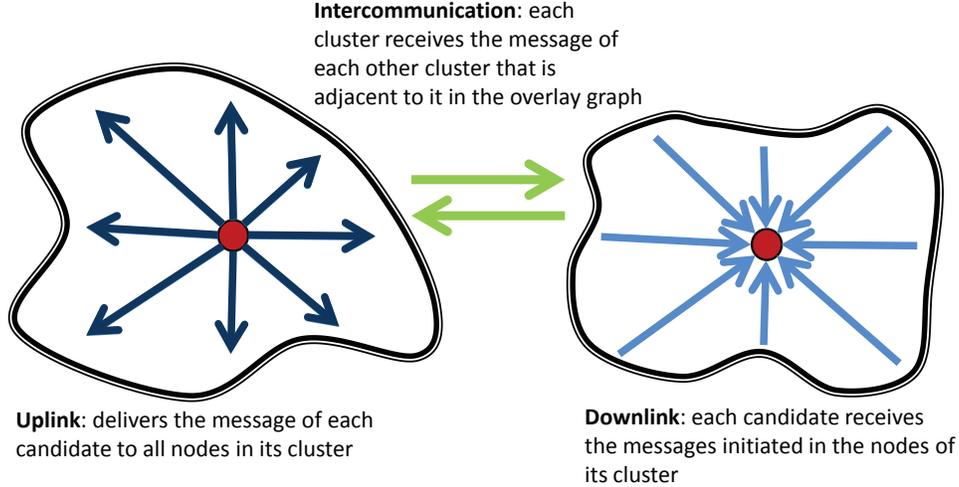}
{\caption{\small A diagram showing the (ideal) communication actions. Each black curve presents a cluster. For simplicity, the nodes inside the clusters are not shown. We note that these are only the ideal definitions of the actions and in the implementations, we will need to modify some of them to match what we can achieve in the given setting. For instance, later the downlink will be changed to receiving just one of the message initiated in the nodes of the cluster, instead of all of them. We then modify the higher level algorithms on the overlay as well to cover for the change.
}}
\label{fig:Communications}
\vspace{-8pt}
\end{figure}
}

We show that both creating clusters and communication actions inside each cluster (uplink and downlink) can be done in broadcast time $T_{BC}$ while intercommunication over borders, which is a local problem, can be solved in (poly-)logarithmically many rounds. With these building blocks simulating one communication round of the overlay network $H$ takes $O(T_{BC}+\poly \log)$ rounds. This is already (almost) our desired final running time. We thus want an algorithm that makes sufficient progress in each debate such that only $O(\log \log n)$ of debates are needed while each debate requires only constant rounds of communication in $H$. We achieve this using the following \emph{Elimination} algorithm, which we run (modulo changes that are explained later) on the overlay graph $H$.

% The algorithm needs to furthermore work with high probability in $n$ (which is exponentially high probability in the number of candidates). 

\paragraph{The Elimination Algorithm:} 
The Elimination algorithm is a simple, deterministic algorithm which makes at least half of the candidates drop out while at least one candidate remains. This algorithm is run by candidates and as a $\mathcal{LOCAL}$ model\footnote{The $\mathcal{LOCAL}$ model\cite{Peleg:2000} is a standard message passing model of distributed computing (in wired networks) where in each round, each node can send one message to each of its neighbors. The elimination algorithm in fact is in the more restricted $\mathcal{CONGEST}$ model\cite{Peleg:2000} where each message size is bounded to $O(\log n)$, and with the additional restriction that the same message is sent to all the neighbors.} algorithm on the overlay graph $H$.

\begin{mdframed}[hidealllines=true,backgroundcolor=gray!35]
%\begin{center}
\vspace{2pt}
%\bigskip 
%\noindent \fbox {\parbox{\columnwidth}{
\textbf{Message Exchange 1:} 
Each candidate sends its ID to all its neighbors and then determines its own degree by counting the number of different IDs received.

\bigskip

\noindent \textbf{Message Exchange 2:} Each candidate sends its ID and its degree to all its neighbors .

\bigskip

\noindent \textbf{Elimination:} Each candidate that is \emph{dominated}---that is, it has a neighboring candidate with a larger degree or a neighbor with equal degree but larger ID---gets \emph{marked} for \emph{elimination} from the candidacy.
%}}
\vspace{2pt}
%\end{center}
\end{mdframed}

%A node $u$ remains unmarked if no neighbor has a higher degree than $u$ and all neighbors with the same degree as $u$
%have lexicographically smaller IDs.  
%\bigskip
\noindent Given these guarantees of the Elimination algorithm, we can infer that a total of $\log (20 \log n)=O(\log \log n)$ debates suffice to reduce the number of candidates from $20 \log n$ to one remaining leader. Since the statement and its proof are simple and instructive, we present them here.

\begin{lemma}\label{lem:slash}
The deterministic Elimination algorithm uses just two rounds of message exchange in the overlay $H$ (between candidates) and eliminates at least half of the non-isolated nodes, while keeping at least one. 
\end{lemma}
\begin{proof} Clearly, the node with maximum $(degree(.), ID(.))$ pair remains. To see that half of the non-isolated nodes are eliminated, we use a potential argument. We give a charge of one to each non-isolated unmarked node. Then we redistribute these charges by each non-isolated unmarked node splitting its charge evenly between its neighbors. Since only non-isolated unmarked nodes initially get charged, and as no two unmarked nodes are neighbors, all charges get redistributed to the marked nodes. Furthermore, each marked node $u$ gets a charge of at most one. This is because, each neighbor $v$ of $u$ gives a charge of $\frac{1}{d(v)}$ to $u$ and we have $\frac{1}{d(v)}\leq \frac{1}{d(u)}$. Here, $d(u)$ and $d(v)$ are the degrees of nodes $u$ and $v$ and we know that $d(v)\geq d(u)$ as otherwise $u$ would be marked. Since each marked node gets a charge of at most one, the total charge is at most as large as the number of marked nodes. Since the total charge was initially equal to the number of unmarked nodes, and since the total charge did not change in the redistribution step, we get that the number of unmarked nodes is at most as large as the number of marked nodes. Thus, the number of marked nodes is at least half the total number of nodes, which completes the proof. 
\end{proof}

\fullOnly{

\paragraph{Debate Template:}Each debate is an implementation of the elimination algorithm on top of the overlay graph $H$. Given the communication primitives uplink, intercommunication, and downlink that are available atop the overlay graph (see their description above), this implementation follows roughly from the outline presented in Algorithm \ref{alg:debateTemplate}. In the following sections, we describe how this debate template can be implemented in each model.

\begin{algorithm}[H]
	\caption{Template of a Debate}
	\label{alg:debateTemplate}
	\begin{algorithmic}[1]
				\algorithmsize
%				\Statex
				\State Cluster \Comment{Overlay Design}
\smallskip
%				\Statex
				\State Uplink candidate IDs \Comment{Exchange 1}
				\State Intercommunicate IDs
				\State Downlink IDs
\smallskip
%				\Statex
				\State Candidates determine their degree in $H$ by counting the number of distinct received IDs
\smallskip
%				\Statex
				\State Uplink pairs of $(degree, ID)$ from candidates \Comment{Exchange 2}
				\State Intercommunicate the pairs
				\State Downlink the pairs
\smallskip
%				\Statex
				\State Each candidate remains iff its $(degree, ID)$ pair is greater than all pairs it receives \Comment{Elimination}
	
	\end{algorithmic}
\end{algorithm}

}

\subsection{Implementation of a Debate Without Collision Detection}\label{sec:overviewnoCD}
Here we present the main ideas for how to implement the aforementioned debate templates in the radio network model without collision detection. The goal is to run one debate in $O(T_{BC} + \log^3 n)$ rounds with high probability and thus obtain the leader election algorithm claimed in \Cref{lem:mainnoCD}. 
\medskip

\noindent\textbf{Clustering:}
There are two basic ways to use the Decay Broadcast strategies for building the clusters. One is to simply run a global broadcast in time $T_{BC}$ using the Fast-Decay and with the candidate IDs as messages. If every node simply keeps and forwards the first ID it receives, in the end, every node belongs to a cluster, and also the clusters are connected. However, the clusters obtained this way do not have nice shapes and do not allow for efficient communication between clusters. The exact issue is somewhat detailed but just to convey the intuition, \Cref{fig:BadClusters} shows a pictorial example. 
This problem can be avoided using the second way, which uses a slower variant of the Decay: we repeatedly use
long-phases of Decay, where in each long-phase, the clusters grow by one hop. Each time all unclustered nodes with a clustered
neighbor get included in the cluster with high probability. This leads to nicely shaped clusters in which each node joins
the closest candidate. But, the running time of this method is $\Theta(D\log^2 n)$ rounds, which we cannot afford. 
\IncludePictures{
\begin{figure}[t]
\centering
\includegraphics[width=0.7\textwidth]{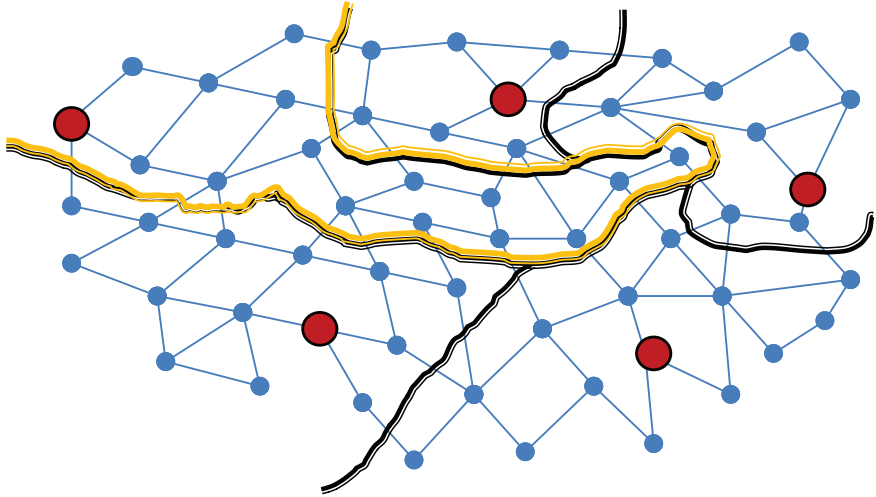}
{\caption{\small A configuration where all nodes are clustered but the clustering is not good for our purposes. The exact issue is somewhat detailed, but here is rough intuition about it: The cluster indicated with an orange curve has a long and thin spike, and all of the nodes on this spike have neighbors in different clusters. Because of this, the communications in the nodes of this spike interfere with the communications of the nodes of the other clusters, and we cannot guarantee fast communication actions inside this cluster.
}}
\label{fig:BadClusters}
\end{figure}
}

Our solution is to combine these two methods to get the best of both worlds. For this we start with a fast-growth 
phase in which we use the first method---namely Fast-Decay---to advance the clusters in iterations of $\Theta(\log^2 n)$ rounds. After each iteration
we ensure that clusters do not interfere with each other, by cutting them back (trimming them) if they do. After the fast-growth phase
which takes $O(T_{BC})$ rounds, clusters are at most $\Theta(\log n)$ far away from each other. Now we use the slower
second method---namely the basic version of the Decay---to grow the clusters carefully spending $\Theta(\log^2 n)$ rounds for each of the remaining $\Theta(\log n)$ steps. This gives
us a nice clustering for a total of $\Theta(T_{BC} + \log^3 n)$ rounds.

\medskip
\noindent\textbf{Overlay Communication:}
Due to the nice clustering, the overlay communication routines for intercommunication and uplink can be easily implemented in
$\Theta(\log^3 n)$ and $\Theta(T_{BC})$ rounds, respectively. However, implementing a downlink is more troublesome.
The subtle reason is that while there are at most $20 \log n$ distinct IDs of neighboring clusters that need to be collected in each candidate, there are copies of each of these IDs registered at up to $O(n)$ different cluster nodes. This prevents classical 
gathering protocols (e.g., \cite{BII93}) to work for this task. 

To remedy this, we use the Fast-Decay broadcast algorithm with time complexity $O(T_{BC})$ within each cluster to inform the candidate about \emph{just one} of its neighbors. After this, the candidate can use the uplink to give feedback to all cluster nodes that it has received this particular ID, again in broadcast time of $O(T_{BC})$ rounds. This guarantees that after this, only nodes with a new piece of information will participate. Repeating this $r$ times results in at least $r$ distinct IDs being learned by the candidate in $O(T_{BC} \;r)$ rounds. While this is an improvement over the naive gathering (which would take $\Theta(n)$ rounds), it still takes a prohibitively large $\Theta(T_{BC} \log n)$ rounds to learn about all neighboring clusters. Next, we show how to work around this issue by modifying the elimination algorithm.

\medskip
\noindent\textbf{The Modified Elimination Algorithm:}
Our modifications to the elimination algorithm are based on the following two ideas: First, we run the elimination algorithm not directly on the overlay graph $H$ but instead on a sparse subgraph $H'$ of $H$. Secondly, we modify the elimination algorithm such that a node needs to be aware of at most $6$ of its neighbors, instead of all $O(\log n)$ of them which was required in the original elimination algorithm. 
\IncludePictures{
\begin{figure}[t]
\centering
\includegraphics[width=0.6\textwidth]{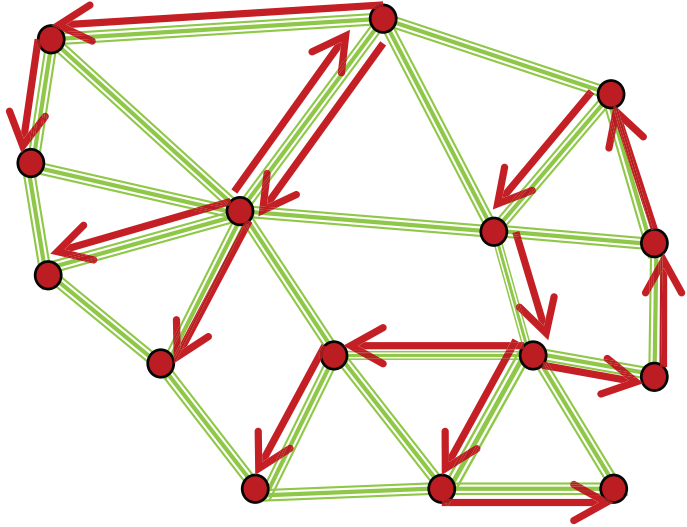}
{\caption{\small An example of the sparsified overlay graph $H'$. The edges of the original overlay graph $H$ are shown as green undirected links. The edges of the sparsified overlay graph are shown as red directed links. Note that even the undirected version of $H'$ can be disconnected. However, any candidate that is not isolated in $H$ will not be isolated in (the undirected version of) $H'$, either.
}}
\label{fig:Sparsified-Overlay}
\vspace{-8pt}
\end{figure}
}
To carve out the sparse sub-graph $H'$, each cluster selects one (incoming) edge from an arbitrary neighbor  and we define $H'$ to be the sub-graph 
consisting of the union of these edges. \Cref{fig:Sparsified-Overlay} shows an example. For the rest of the high level discussion here, the reader can imagine this graph to be undirected. The details of how the directions are treated can be found in \Cref{sec:overlaynoCD}. Note that, in the undirected version of $H'$, although the average degree of a node is at most $2$ in this graph, this does not
hold true for the maximum degree. Nonetheless, using the inward-communication scheme explained above with $r=5$, each candidate can learn about
all of its neighbors if it has at most $5$ of them, and at least detect that it has five or more neighbors otherwise. With this knowledge,
we run the same elimination algorithm as before except for the modification that, any node with degree of at least $5$ remains unmarked. The reason for this is because these nodes cannot safely determine whether their degree is dominated by a neighbor. Fortunately however there are at most a $\frac{2}{5}$-fraction of nodes with degree of at least five since more would lead to an average degree of more than $2$. The Modified Elimination Algorithm therefore still eliminates at least a $\frac{1}{2} - \frac{2}{5} = \frac{1}{10}$ fraction of the candidates, while remaining safe.

%This suffices to implement each debate in $O(T_{BC} + \log^3 n)$ rounds while keeping the number of debates to at most $O(\log \log n)$.

\subsection{Implementation of a Debate via Beeps}\label{sec:overviewCD}

%\bcomment{This Section is work in progress}

In this section we describe the main ideas for implementing a debate in the beep model (or radio network model with collision detection). Our algorithm works along the lines of the debate template presented in \Cref{sec:template}: It first clusters the nodes and then uses overlay communication protocols to run the elimination algorithm.
 
We first introduce our main tools, \emph{beep waves} and \emph{superimposed codes}, and explain how to use them to cluster the graph and implement the overlay communication protocols mentioned in \Cref{sec:template}. We then put everything together and present a simple debate implementation that runs in $O(D + \log^3 n)$ rounds. Lastly, we show how to achieve the running time claimed in \Cref{lem:mainCD} by modifying the simple debate implementation to run in $O(D + \log n \log \log n)$ rounds.

\medskip
\noindent\textbf{Beep Waves: }%\subsubsection{Beep Waves}\label{sec:overviewbeepwaves}
The main difference between radio models without collision detection and those with collision detection (or beeping) is the ability to create what we call \emph{beep waves}. Beep waves start at one or more nodes by sending a beep, and after this initiation, each node that hears a beep forwards it by beeping in the following round. This way, the beep propagates in the form of a wave throughout the network, moving exactly one step per round. \Cref{fig:Onewave} shows an example.

\IncludePictures{
\begin{figure}[t]
\centering
\includegraphics[width=0.5\textwidth]{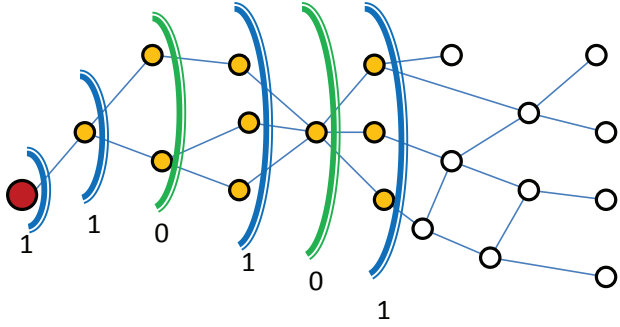}
{\caption{\small A snapshot of an example beep wave that carries the bit string 101011. The wave has started from the red node and is moving towards the right, with a speed of one hop per round.
}}
\label{fig:Onewave}
\vspace{-8pt}
\end{figure}
}

Beep waves have several applications. For one, they can be used to determine the distance of a node $u$ from a (beep) source, by measuring the number of rounds for the wave to reach $u$. Secondly, pipelining multiple beep waves from a source can be used to transmit bit strings, by coding $1$ into a beep and $0$ into absence of a beep. Pipelined beep waves will be our main tool in implementing the communication protocols used in our leader election algorithms of the beep model. 

\medskip
\noindent{\textbf{Superimposed Codes: }}%\subsubsection{Superimposed Codes}\label{sec:overviewcodes}
Another interesting feature of using beep waves to transmit information is that, when two different sources $s_1$ and $s_2$ simultaneously send different bit strings to one node $v$ with equal distance from $s_1$ and $s_2$, then $v$ receives the superimposition or bit-wise OR of the two strings. Typically, such a bit string is considered useless. Thus, protocols designed for radio networks so far have mostly focused on using collision detection and randomization to detect and avoid collisions. In this paper, we take the exact opposite stance: instead of avoiding collisions, we propose to embrace them and leverage their superimposition nature. The key element we use for this is (variants of) \emph{superimposed codes}. These codes consist of codewords that allow any superimposition of a bounded number of codewords to be decomposed and decoded. We note that the general class of superimposed codes are old concepts, dating back to the $40$'s\cite{mooers1948application}. As indicated in \cite{indyk1997SIcodes}, there are several variants of the superimposed codes and they have many different applications. A canonical application of these codes is to assign a signature to each document, by superimposing the codewords assigned to the terms in the document\cite{knuth1973vol, faloutsos1984signature}. We also note that superimposed codes have been used in multi-hop radio networks before, e.g., \cite{CMS01}, where they are viewed as \emph{selective families} and used to schedule successful transmissions by avoiding collisions. Our application is quite different as instead of trying to avoid collisions, we intentionally use collisions for conveying information and use superimposed codes to extract information from these collisions\footnote{In this aspect it is similar to \cite{BCC12}. However there an additive interference of signals in a finite field is assumed which makes collisions completely reversible which allows the efficient use of network coding methods \cite{NC2011}. Exploiting such negative interference is however not possible using only collision detection. This makes the approach here is much more applicable and relevant.}
%Also, in our final algorithm, we just use a simple and very-weak version of these codes, just to approximate the number of the codewords involved in the superimposition.

 %Interestingly, the construction, existence and use of superimposed codes itself turns out to be a very old concept from the $40$'s, e.g., where they were used as an efficient way to use punch-card archival systems.

\IncludePictures{
\begin{figure}[t]
\centering
%%\begin{tabular}{cc}
%\vspace{-0.9cm}
%\begin{minipage}{3in}
%\includegraphics[width=2.7in]{two-bit-waves-small2.pdf}
%%\end{minipage}&
%\end{minipage}\\%[0.5cm]
%\vspace{-1.5cm}
%\begin{minipage}{3in}
%\includegraphics[width=2.7in]{two-bit-waves-big2.pdf}
%\end{minipage}
%%\end{tabular}
\includegraphics[width=\textwidth]{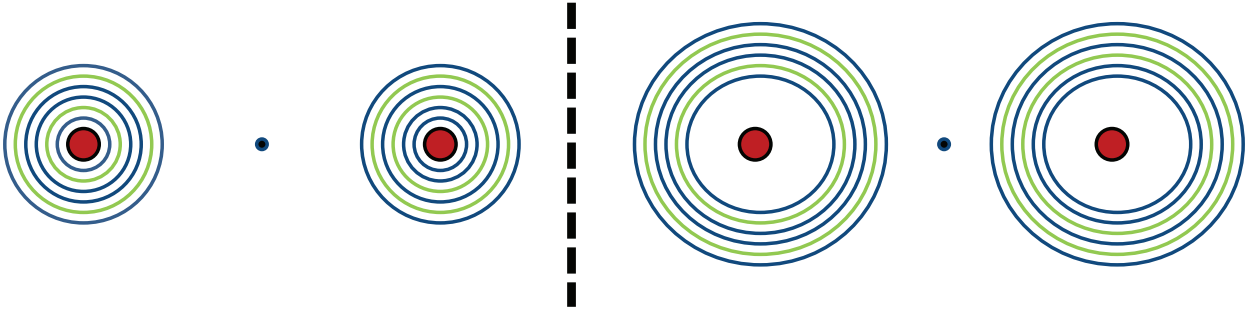}
{\caption{\small A pictorial example of two nodes sending bit strings encoded in beep waves. Both sides of the figure show the snapshots of the same beep waves; the right side is a later snapshot after the waves have spread to larger radii. Sending a bit 1, which is transmitting a beep, is indicated with blue circles, and sending a bit 0, which is listening (not beeping), is indicated with green circles. That is, for example, in the snapshot, all the nodes on a blue circle are beeping, which means they are transmitting a 1. The left candidate is sending a beep wave which carries the bit pattern 101101, while the right candidate is sending 101011. The right figure shows the snapshot of the waves at a later time. The node in the middle, which is at the exact same distance from the two candidates, will receive the superimposition of the waves, i.e., the bitwise or of the two bit strings, which is equal to 101111. If a superimposed code is used and the aforementioned bit patterns are the results, the middle node would be able to decode the received bit pattern 101111 into two separate initial bit patterns 101101 and 101011.}}
\label{fig:beepwaves}
\end{figure}
}

\medskip
\noindent\textbf{Clustering and Overlay Communication: }%\subsubsection{Clustering and Overlay Communication}\label{sec:overviewbeepclustering}
To \emph{cluster} nodes around candidates, we assign each node to the closest candidate if there is a unique such candidate, and leave all other nodes unclustered. To this end, we first use beep waves from each candidate to determine for each node its distance to the closest candidate. Using these distance numberings, we then send out the IDs of the candidate nodes as pipelined beep waves. To prevent nodes confusing superimposition of IDs with clean IDs, we use a superimposed code that allows to distinguish whether a received bit string is a coded ID or the superimposition of multiple coded IDs. This allows nodes with more than one closest candidate to stay unclustered while nodes with a unique closest candidate join the cluster of that candidate. This clustering also induces the overlay graph $H$.
\smallskip 

Next, we implement our uplink, intercommunication, and downlink \emph{communication protocols} that allow communication atop the overlay graph $H$.  We implement these by using the distance numbering to synchronize pipelined beep waves for communication. This is relatively straightforward but leads to both downlink and intercommunication not actually delivering all messages from all clustered nodes or neighboring clusters respectively, but instead delivering the superimposition of all these messages. This is where superimposition codes show their full power. Instead of sending messages directly in the intercommunication phase, we use messages coded with a superimposition code. When a candidate then receives the messages from its adjacent candidates, in superimposition after they got combined in the intercommunication and downlink phase, it can still fully reconstruct all original messages. This way superimposition codes allow us to implement a full local message exchanges in $H$ even though the actual intercommunication and downlink protocols merely deliver superimpositions. 

\medskip
\noindent\textbf{Implementing the Elimination Algorithm Optimally:} %\subsubsection{Implementing the Elimination Algorithm Optimally}
Given the full message exchange overlay communication, one can directly implement the debate template of \Cref{sec:template}. The running time of such a debate would be $O(D + \log^3 n)$. Here the $\log^3 n$ comes from the length of superimposed codes that allow to decode messages from the superimposition of up to $\log n$ codewords which we use to implement the full message exchange over $H$. In the remainder of this section, we show how we can improve over this to achieve a debate with round complexity of $O(D + \log n \log \log n)$. 

\smallskip

The main observation that leads to the speedup is as follows: The $O(D + \log^3 n)$ debate implementation does not actually use the full power of message exchange that is given by the use of the $O(\log^3 n)$ bits superimposed codes. Instead, we eventually use this communication only for two tasks: (1) determining the number of different messages received; used for determining the degree on the overlay graph and the cluster boundaries, and (2) checking whether there is a neighbor with a larger ($degree$, $ID$) string; used to decide whether a node marks itself in the elimination algorithm. We show that both tasks can be achieved with a smaller overhead.
\smallskip

For (1), we design a new set of codes that are just strong enough to enable us to estimate the number of codewords in a superimposition up to a multiplicative $(1+\delta)$-factor, for any constant $\delta>0$. For instance, this can be used in estimating the degree of each candidate up to a factor of $1+\delta$. These new codings encode each ID of length $\log n$ into a codeword that is only $\Theta(\log n \log \log n)$ bits large. We then show that eliminating candidates with the Elimination algorithm based on these approximate degrees still works, that is, still removes a constant factor of candidates per debate while keeping at least one. This can be easily checked by following the same potential argument as in the proof of \Cref{lem:slash}, noting that now, each remaining node gets a charge of at most $4$ (assuming e.g. $\delta=1$). With this, the number of remaining candidates is at most $4$ times larger than the number of marked candidates. In other words, still at least a $\frac{1}{5}$ fraction of the candidates gets removed.
\smallskip

For usage (2), that is detecting whether a neighboring candidate has a message numerically larger or not, we use a slightly different intercommunication algorithm, which we explain next: Nodes go through the bits of messages of their candidates one by one and compare them. They \emph{mark} themselves if they detect a larger message in the neighborhood. A node that gets marked does not continue the process anymore. It is easy to see that a node gets marked in this procedure if there is an adjacent cluster with a larger message. Finally, to deliver this information to the candidates, we simply use a Downlink with single-bit messages (marked or not) and each candidate gets to know whether any of nodes in its cluster is marked or not.

\subsection{Obtaining Linear Time Leader Election Algorithms}\label{sec:overviewlinear}
In this section, we explain a simple optimization which reduces the multiplicative $\log \log n$ factor in our bounds for networks with (near) linear diameter. In particular, this optimization makes all our running times $O(n)$. 

While we do not know how to reduce the number of debates below $\Theta(\log \log n)$, we show that less time can be spent on initial debates. That is, if there are many candidates, we can work with clusters that have a smaller diameter. In particular, we note that in both of the leader election algorithms presented above, we have the following property: In each debate, the time used for growing clusters is large enough such that the radius of each cluster can potentially grow up to $D$. This way, the algorithms avoid isolated clusters. Note that isolated clusters do not get eliminated in the Elimination Algorithm (see \Cref{lem:slash}). Furthermore, the time needed for a debate depends directly on this \emph{radius of growth}. In particular, in the model without collision detection, growing clusters up to radius of $d$ takes $O(d \log \frac{n}{d} + \log^3 n)$ and in the beeping model, it takes $O(d + \log n\cdot \log\log n)$.

%In this, radius of a cluster is the maximum distance of a node from the respective candidate in the induced graph on the cluster. On the other side, the time complexities of debates depend directly on this distance of growth. In particular, in the model without collision detection, growing clusters up to radius of $d$ takes $O(d\log \frac{n}{d}+ \poly\log)$ and in the beeping model, it takes $O(d+ \poly\log)$. Thus if in each of the debates, we grow the clusters for potential radius up to $D$, then the time complexities of the whole leader election algorithms are already $O(D \log{\frac{n}{D}} \log \log n)$ and $D \log \log n$, respectively for the two models. Note that in worst case of $D = \theta(n)$, these time bounds are $O(n \log \log n)$.

The key observation is as follows: if the number of remaining candidates is $k$ then at most half of the candidates can be such that they have no other candidate within distance $2n/k$. Because of this, the idea is that in the $i^{th}$ debate, instead of building clusters for radius up to $D$, we grow the clusters for radius only up to $\min\{D, \frac{4n\; 1.05^i}{\log n}\}$.
This still ensures that at least half of the candidates are non-isolated, which allows the Elimination Algorithm to remove at least a constant fraction of all (non-isolated) nodes. It is an easy calculation to see that this change in radius of growth reduces the $\log \log n$ factor in our time bounds to a $\min\{\log \log n,\log  \frac{n}{D}\}$ factor, as claimed in \Cref{lem:mainnoCD} and \Cref{lem:mainCD}. The details for this are given in \Cref{sec:linearnoCD} and \Cref{subsec:linearLEWCD}.

% We defer the details to the full version. %In particular, this makes our algorithms to run in optimal $O(n)$ rounds. 

\fullOnly{
\section{Leader Election without Collision Detection}
This section is devoted to providing the technical details and proofs for \Cref{lem:mainnoCD}. As described in \Cref{sec:overview}, this algorithm follows the template given in \Cref{sec:template} and uses the ideas explained in \Cref{sec:overviewnoCD} to implement this template. In particular, we present the clustering algorithm in \Cref{sec:clusteringnoCD}, show how to obtain the sparsified overlay graph $H'$ and perform overlay communication protocols on top of it, in \Cref{sec:overlaynoCD}, and then, we explain in \Cref{sec:debatenoCD} how to implement a debate in $O(T_{BC} + \log^3 n)$ time. Finally, in \Cref{sec:linearnoCD}, we explain how to add the simple optimization trick described in \Cref{sec:overviewlinear} to complete the proof of \Cref{lem:mainnoCD}.

%lead to a $\tilde{O}(T_{BC} + \log^3 n)$ leader election algorithm as claimed in \Cref{lem:mainnoCD}.
%After this, in \Cref{sec:linearnoCD}, we provide the proofs for optimizing the algorithms to linear time for networks with large diameters that were discussed in \Cref{sec:overviewlinear}. These details together lead to a $\tilde{O}(T_{BC} + \log^3 n)$ leader election algorithm as claimed in \Cref{lem:mainnoCD}.

\subsection{Clustering}\label{sec:clusteringnoCD}

In the clustering phase, we partition the network into disjoint \emph{clusters}, one around each candidate, such that these clusters provide a platform for easy communications between the candidate nodes. 

Formally, a \emph{clustering} is a partial assignment which for each node $v$, it either \emph{clusters} $v$ by assigning it to (the \emph{cluster} of) a candidate, or it leaves $v$ \emph{unclustered}. Given a clustering, we say a clustered node $v$ is a \emph{boundary} node if $v$ has a neighbor $u$ that is unclustered or belongs to a different cluster. Otherwise, we say the clustered node $v$ is an \emph{internal} node. However, as a small exception, candidates themselves are considered as internal always. We maintain the invariant that each clustered node is connected to the candidate it is assigned to via a path of internal nodes of the same cluster. See the first part of \Cref{fig:cluster-refinement} for an example of a single cluster, with its candidate, internal and boundary nodes.

For a given distance parameter $d$, we say that two clusters $\mathcal{C}_1$ and $\mathcal{C}_2$ (or their respective candidates) are \emph{$d$-adjacent} if there are internal nodes $v_1 \in \mathcal{C}_1$ and $v_2 \in \mathcal{C}_2$ such that $v_1$ and $v_2$ are within distance at most $d$ of each other. This notion of $d$-adjacency defines the \emph{candidate} graph $H_d$, where two candidates are adjacent in the overlahy $H_d$ if they are $d$-adjacent. We say a clustering has \emph{connectivity gap} at most $d$ if the graph $H_d$ has no isolated node (or exactly one isolated node if that is the only node of the graph $H_d$). That is , if each cluster is $d$-adjacent to at least one other cluster (assuming that there are more than one clusters).

We would like to obtain a clustering with a small connectivity gap, ideally just a constant gap. This is because, the intercommunication action---which carries the message of different clusters over this gap and delivers it to other adjacent clusters---has a round complexity that is monotonically increasing with (and almost directly proportionally to) the connectivity gap. We next explain how to obtain a clustering with constant connectivity gap. Later, we use this clustering for building the desired \emph{sparsified overlay} graph $H'$ and for communications over $H'$. 
%We see that this constant connectivity gap translates to a $O(\log^3 n)$ intercommunication time complexity. 

%\iffalse
%\smallskip
%
%Growing clusters is done in two parts, namely \emph{fast growth} and \emph{slow growth}. In a nutshell, in the fast growth part, we grow the clusters quickly so that (1) each cluster has at least one other cluster within its distance $O(\log^2 n)$, and (2) each boundary node is connected to its respective candidate via a path made of only internal nodes of that cluster. Then, in the slow growth part, we grow the clusters with a slower peace and in a more steady manner, to close up the $O(\log^2 n)$ gap between the clusters, achieving the properties ($*$) and ($**$) mentioned above. 
%\fi

To achieve a clustering with constant connectivity gap, we start from a trivial clustering where each candidate is one cluster. This trivial clustering has a connectivity gap of at most $D$. We then reduce the gap by using an algorithm which we call Fast-Cluster. This algorithm uses $O(T_{BC})$ rounds and produces clusters with connectivity gap of $O(\frac{\log^2}{\log {\frac{n}{D}}})$, with high probability. After that, we use a simpler algorithm that we call \emph{Cluster-Refinement} which essentially performs $\Theta(\frac{\log^2 n}{\log {\frac{n}{D}}})$ long-phases of Decay to reduce the connectivity gap to a constant. In what follows we first give the Fast-Cluster algorithm and its analysis and then provide and analyzed the Cluster-Refinement algorithm.

\newparagraph{The Fast-Cluster Algorithm:} This algorithm uses $\Theta(T_{BC}) = \Theta(D\log \frac{n}{D}+\log^2 n)$ rounds, which are divided into \emph{epochs} of $\Theta(\log^2 n)$ rounds each. At the beginning and the end of each epoch, we will (w.h.p.) have valid clusterings in which: (1) each clustered node knows the ID of the candidate it is assigned to, (2) each node is either ``internal'', ``boundary'' or ``unclustered'', (3) each clustered node is connected to the candidate of its cluster via a path of internal nodes of the same cluster.
 
%Initially, each cluster only contains the related candidate. 
During each epoch, instead of the status ``boundary'', we have a different status called \emph{``undecided''}, which indicates that a node is temporarily clustered. At the end, each undecided node might become ``internal'', ``boundary'' or ``unclustered''. At the beginning of each epoch, we first change the status of every ``boundary'' node to ``undecided''. Then, only the candidates, the unclustered nodes and the ``undecided'' nodes participate in the transmissions of the epoch. The candidates always remain internal, regardless of what happens.

\smallskip

Each epoch consists of four steps. In a nutshell, the first step is for quickly growing the clusters using the Fast-Decay protocol whereas the other three steps are for trimming the resulted clusters, in order to refine the shape of the newly grown parts of the clusters. The key goal in this trimming is to resume the property that each clustered node is connected to the candidate of its cluster via a path of internal nodes of the same cluster. A pictorial example is shown in \Cref{fig:cluster-refinement}. In the first step, we grow the clusters, thus creating more ``undecided'' nodes. Each undecided node is (temporarily) in one cluster. Then, in the second and third steps, we \emph{mark} some ``undecided'' nodes using a particular rule and we use these marks in the fourth step to determine the new statuses. In particular, this marking is used for trimming the clusters and the unmarked undecided nodes will become internal at the end of the epoch. The details of these steps are as follows: 
\medskip

\IncludePictures{
\begin{figure}[h!]
\centering
\begin{minipage}{0.9\textwidth}
\centering
\includegraphics[width=0.7\textwidth]{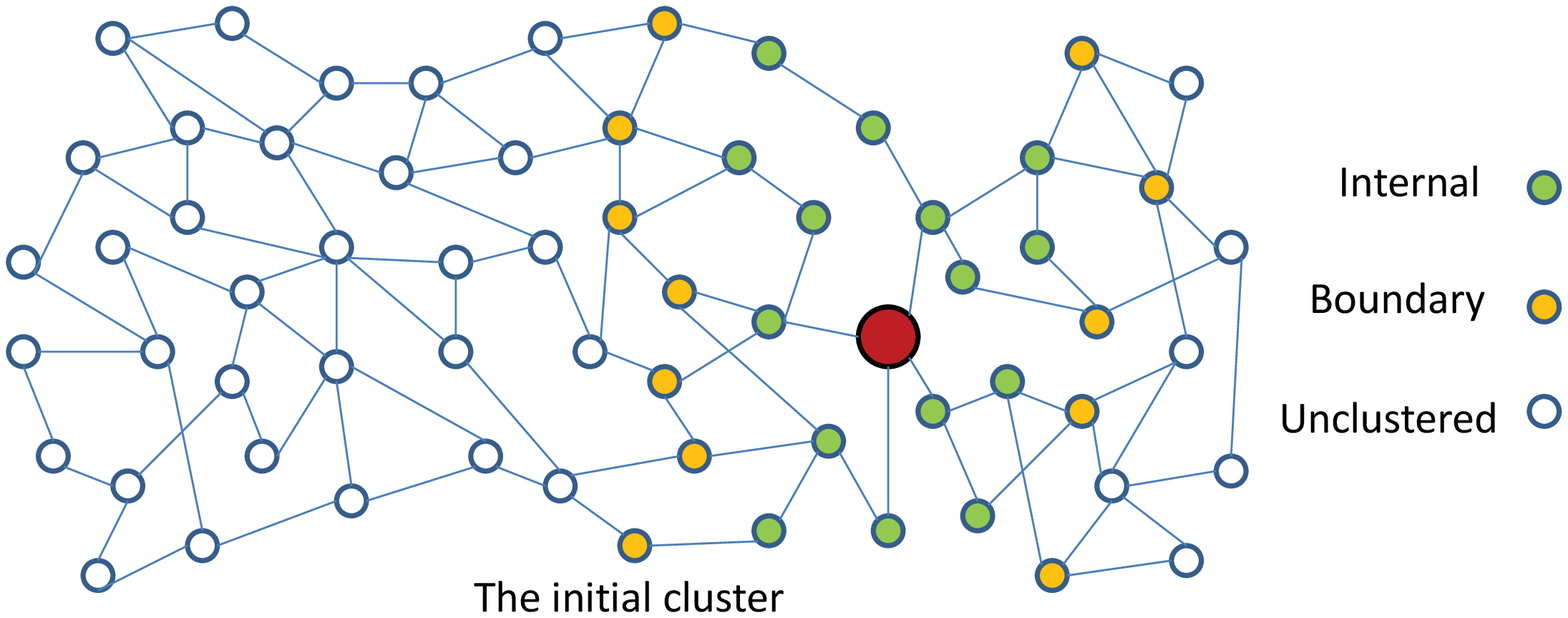}
\end{minipage}\\[0.7cm]
\begin{minipage}{0.9\textwidth}
\centering
\includegraphics[width=0.7\textwidth]{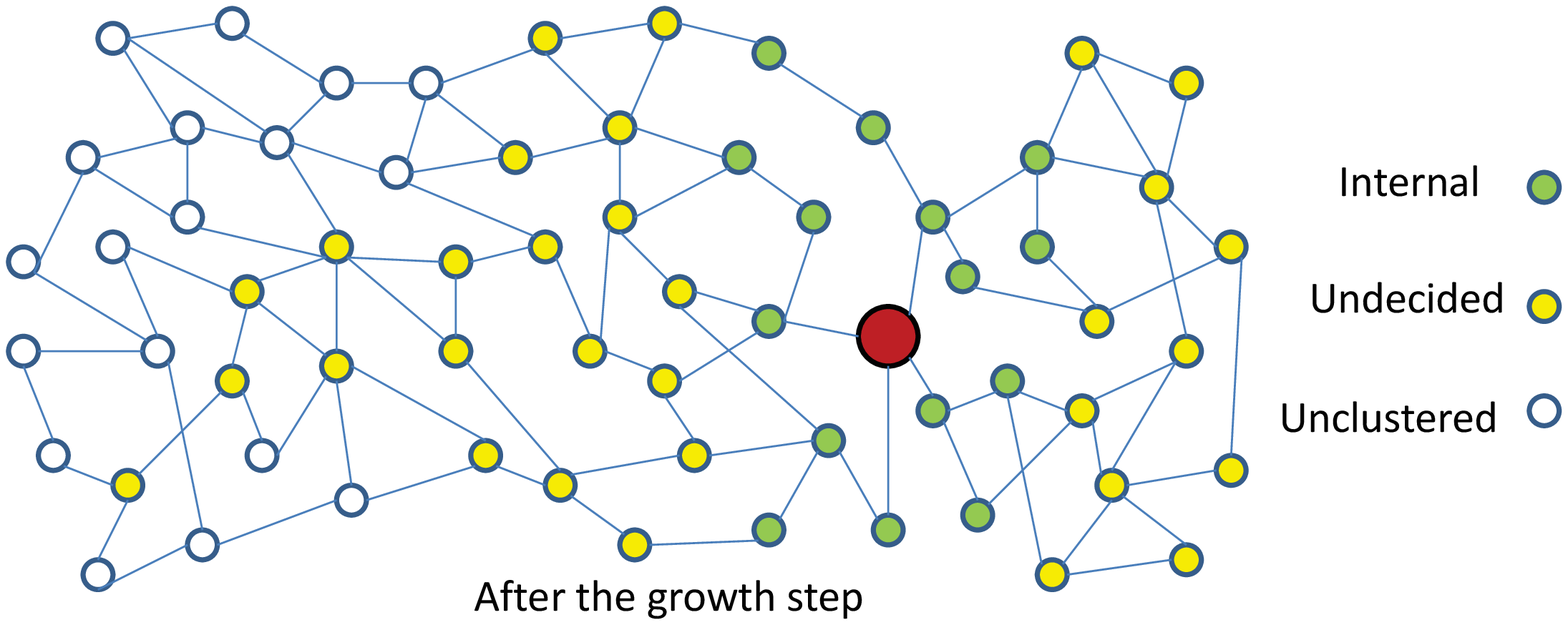}
\end{minipage}\\[0.7cm]
\begin{minipage}{0.9\textwidth}
\centering
\includegraphics[width=0.7\textwidth]{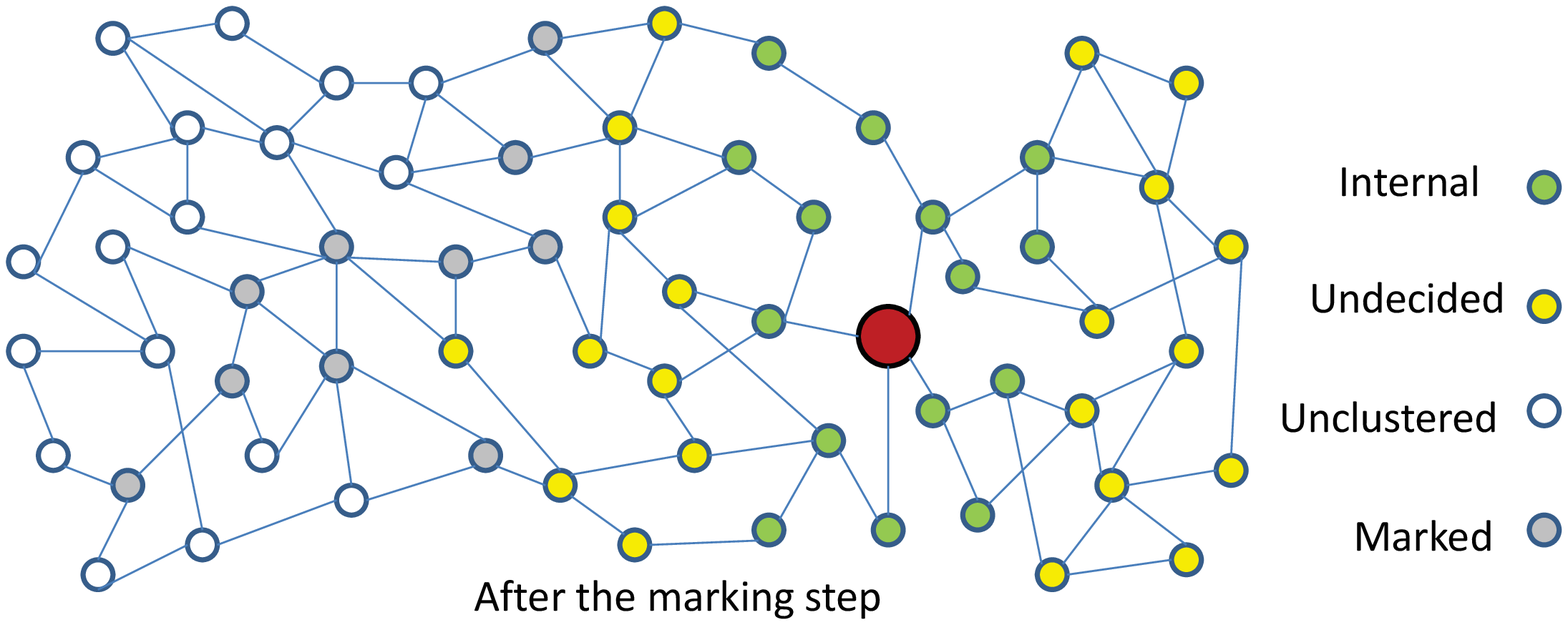}
\end{minipage}\\[0.7cm]
\begin{minipage}{0.9\textwidth}
\centering
\includegraphics[width=0.7\textwidth]{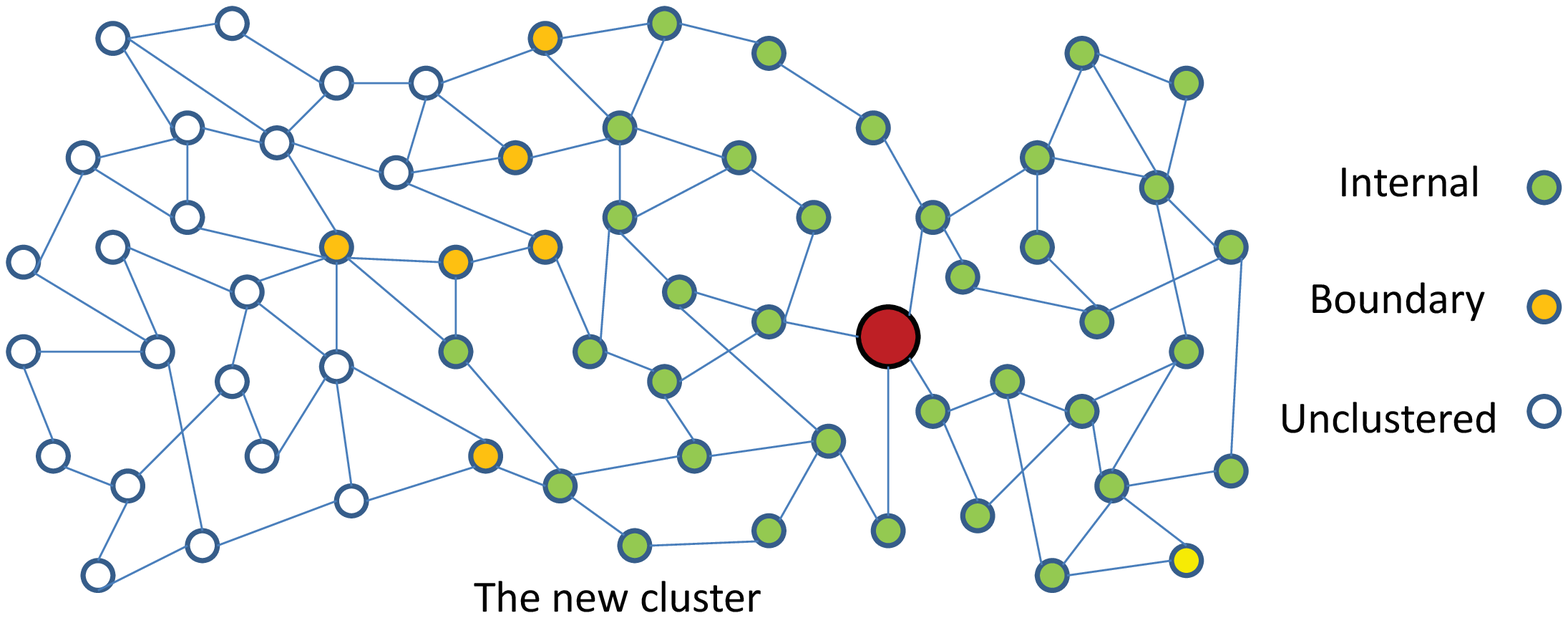}
\end{minipage}

{\caption{\small The growth of a cluster through different steps of the algorithm. For simplicity, the example is depicted only for one cluster. The first graph shows the initial configuration of the cluster. The second graph shows the configuration after Step 1, where a number of undecided nodes are added to the cluster, using some rounds of the Fast-Decay broadcast algorithm. The third graph shows the configuration afters Steps 2 and 3, where some nodes are marked. In particular, the long left-going spike is marked. The fourth graph shows the new cluster, at the end of step 4.}}
\label{fig:cluster-refinement}
\end{figure}
}

\begin{adjustwidth}{1.6em}{0em}
\begin{enumerate}
\item[Step 1:] This step consists of $4 \alpha \log^2 n$ rounds\footnote{Recall that $\alpha$ is the large enough constant used in the definition of the Fast-Decay protocol. See \Cref{subsub:fastDecay}.}. Let $\delta = \log \frac{n}{D}$.  In these rounds, we grow the clusters using the Fast-Decay($\delta$) protocol. For this, each candidate and each ``undecided'' node %that is temporary-boundary or permanent-boundary 
starts with the message of the related candidate. With these messages, we run $4\alpha \log^2 {n}$ rounds of the Fast-Decay($\delta$) protocol. In these rounds, some unclustered nodes receive the message of one or possibly more candidates. Each node ignores all the received messages but the first one. Then, each unclustered node that received some message temporarily joins the cluster of the candidate whose ID is mentioned in that message and changes its status to ``undecided''. 
 
\item[Step 2:] Here, we \emph{mark} any ``undecided'' node that is adjacent to either an unclustered node or a node from a different cluster. This step uses time equivalent to $\Theta(\log n)$ phases of the basic Decay, i.e., $\Theta(\log^2 n)$ rounds.
%using $4 \log n$ short phases of length $T = 30 \log n$ and two long phase of length $T' = 30 \log^2 n$. 
%
In these rounds, each unclustered node simply runs $\Theta(1)$ long-phases of the basic Decay protocol, sending a message declaring that it is unclustered. An undecided node that receives a message from an unclustered node becomes marked. On the other hand, from the viewpoint of the clustered nodes, these rounds are divided into $\Theta(\log n)$ \emph{parts}, each part consisting of one phase of the basic Decay protocol. In each part, a subset of clustered nodes are active. More precisely, in each part, nodes of each cluster \emph{unanimously} decide to be active or be listening, each with probability $\frac{1}{2}$. Note that this unanimous cluster decision can be achieved by each candidate sharing $\Theta(\log n)$ bits of randomness with nodes of its cluster by attaching these randomness bits to its ID (so the size of the messages remain asymptotically the same). In each part, all active nodes then perform one phase of the basic Decay algorithm, sending their cluster ID.  
%An undecided node $v$ that receive a message from containing a cluster ID different than the cluster than contains $v$ or a message from an unclustered node becomes marked.
All clustered nodes that receive a message different than their own cluster ID become marked. 
%
%In the first long phase we perform $T' = 30 \log^2 n$ rounds of Decay($T'$) with the ``unclustered'' nodes marking all clustered nodes that receive a message. This marks all ``undecided'' nodes neighbor an unclustered node or a node from a different cluster. To achieve distance two we perform $T' = 30 \log^2 n$ rounds of Decay($T'$) in the second long phase with all marked nodes marking all clustered nodes receiving a message.

\item[Step 3:] In this step we mark the ``undecided'' nodes that were (indirectly) recruited in Step 1 by nodes which got marked in Step 2. To make this more precise, we say a node $v$ was \emph{directly recruited} by node $u$ if in Step $1$, the first message that node $v$ received (which is also the message that resulted in the status of $v$ to become ``undecided'') was received directly from node $u$. Similarly, we say node $v$ was \emph{indirectly recruited} by node $u$ in Step $1$ if the first message that node $v$ received came from a node that was directly or indirectly recruited from $u$. To achieve the marking of nodes that got (indirectly) recruited by marked nodes we repeat the exact same transmissions of Step 1---thus in $\Theta(\log^2 n)$ rounds---but now, each node also adds a bit indicating whether it is marked or not to its message. Any node that receives a set bit from the node that recruited it in Step $1$ becomes marked. 

\item[Step 4:] In this step we determine the final statuses. For this, all non-marked ``undecided'' nodes become ``internal''. Then, we run one long-phase of the basic Decay protocol where the internal nodes transmit and each non-internal node that receives a message becomes ``boundary''. Lastly, we set the status of any remaining ``undecided'' node to be unclustered.
\end{enumerate}
\end{adjustwidth}
Now we analyze this algorithm and show that it achieves a connectivity gap of $O(\frac{\log^2 n}{\log {n}/{D}})$, with high probability. 

As a warm up, we first study the simpler hypothetical scenario where there is only one candidate. We particularly investigate the growth of one cluster in the absence of the others.
\begin{proposition}\label{lem:growth} Consider a starting clustering where there is only one cluster $\mathcal{C}_u$. Then, in each epoch, $\mathcal{C}_u$ grows by at most $\frac{4\alpha\log^2 n}{\log{n}/{D}}$ hops. Furthermore, in each epoch, the growth is at least as much as the growth of the Fast-Decay($\delta$) broadcast that is run (alone) for $\alpha \log^2 n$ rounds. Hence, running it for $\Theta(T_{BC})$ rounds corresponds to a full run of the Fast-Decay($\delta$) broadcast algorithm and means that the cluster grows by $D$ hops, with high probability. 
\end{proposition}
\begin{proof}
Consider an epoch. First note that only step 1---which has $4\alpha\log^2 n$ rounds---can lead to growth in the clusters and because this step uses the Fast-Decay($\delta$) protocol which has a delay parameter $\delta=\log n/D$, in each epoch each cluster can grow at most $\frac{4\alpha\log^2 n}{\log n/D}$ hops. 

Now note that if node $v$ was labeled ``undecided'' during the first $\alpha \log^2 n$ rounds of the Fast-Decay($\delta$) protocol in Step 1, then $v$ becomes ``internal''. This is because of the following: From \Cref{lem:fastdecay-slowstep}, it follows that since we have assumed that there is no other cluster, by the end of the $4\alpha \log^2 n$ rounds of the Fast-Decay($\delta$) protocol used in step 1, any node that is within distance $2$ of $v$ will have received a message, and thus become clustered (at least temporarily as ``undecided"). After that, only nodes that are at distance exactly one or two from $v$ will get marked in Step 2, while $v$ does not get marked. Thus, $v$ remains unmarked and becomes internal. Therefore, in every epoch, the growth of the cluster $\mathcal{C}_u$ is at least as much as the growth of the Fast-Decay($\delta$) broadcast that is run (alone) for $\alpha \log^2 n$ rounds. It follows that running the Fast-Cluster algorithm for $\Theta(T_{BC})$ rounds corresponds to a full run of the Fast-Decay($\delta$) broadcast algorithm. \Cref{lem:fastdecay} then shows that the cluster grows by $D$ hops, with high probability. 
\end{proof}
Now, to understand the effects that different clusters can have on each other, we study the markings done in step 2. Note that a marking might be due to a neighboring clustered node from a different cluster or due to a neighboring unclustered node.
\begin{proposition}\label{prop:marks} If at the start of step 2 of an epoch, an undecided node $w$ is adjacent to an unclustered node or a clustered node of a different cluster, then, with high probability, $w$ gets marked.
\end{proposition}
\begin{proof} Suppose that $w$ is adjacent to an unclustered node $v$ or a node $v'$ clustered in a different cluster. We say a part is good if the the cluster of $w$ is listening in this part and following  holds: either $w$ has an adjacent unclustered node $v$ or there is a clustered node $v''$ adjacent to $v'$ that is from a different cluster and that cluster is active in this part. For each part, the probability that this part is good is at least $1/4$. Moreover, there are $\Theta(\log n)$ parts in total. Hence, with high probability, there are at least $\Theta(\log n)$ good parts. These constitute equivalent of one long-phase of decay where $w$ and all nodes of its cluster are silent and an unclustered neighbor $v$ of $w$ or a clustered neighbor $v''$ of $w$ from a different cluster is running a phase of the basic Decay protocol. Hence, with high probability, $w$ receives a message from one such neighbor and thus gets marked. 
\end{proof}

Now we argue that at the end of each epoch, we have a clustering in which each node can reach its candidate using only internal nodes of the same cluster:
\begin{proposition}At the end of each epoch of the Fast-Cluster algorithm, with high probability, for each cluster $\mathcal{C}_u$, we have the following property: each clustered node $v$ of $\mathcal{C}_u$ is connected to $u$ via a path made of only internal nodes of $\mathcal{C}_u$.
\end{proposition}
\begin{proof}We first show that every \emph{internal} node has a path of internal nodes connecting it to the related candidate. This is true initially and remains true because the only way a node $v$ becomes internal is if it receives a message in Step 1 over a sequence of transmissions starting from a previously boundary node. Furthermore, for the node $v$ to become internal, none of the nodes that recruited $v$ directly or indirectly can be marked in Step 2. This is because otherwise $v$ would also get marked in Step 3. Once node $v$ becomes internal, all nodes connecting it to a node that was internal before (all those that recruited $v$ directly or indirectly) are getting an internal status too, and this preserves the invariant. Now, it is easy to see that the claimed property---that their are connected to their candidate via a path of internal nodes---also holds for boundary nodes created in Step 4 since these boundary nodes are recruited only by internal nodes. 
\end{proof}

Now we are ready to complete the analysis of the Fast-Cluster algorithm, by proving its connectivity gap guarantee.
\begin{lemma}\label{lem:clustergrowth}
The algorithm Fast-Cluster produces a clustering with connectivity gap $O(\frac{\log^2 n}{\log n/D})$, with high probability.
\end{lemma}
\begin{proof}[Proof of \Cref{lem:clustergrowth}] 
Consider a candidate node $u$ and its cluster $\mathcal{C}_u$ after running the algorithm. To prove the lemma, we prove that either $u$ is the only cluster or there is another cluster $\mathcal{C}_{u'}$ that is $d'$-adjacent to $\mathcal{C}_{u}$. This shows that the new connectivity gap is at most $d'$.

We know that either there is no cluster other than $\mathcal{C}_u$ or there is at least one other cluster that is $D$-adjacent to $\mathcal{C}_u$. The former case is as studied in \Cref{lem:growth}. We now consider the latter case, which is the more interesting one. In this case, there will be a time in which the growth of the cluster $\mathcal{C}_u$ interferes (and is thus slowed down) by that of the other clusters. Following the arguments in the proof of \Cref{lem:growth}, we see that during a complete run of the Fast-Cluster, w.h.p., there will be a round in which one ``undecided'' node of $\mathcal{C}_u$ gets marked for neighboring a node from another cluster in Step 2. Suppose that $i$ is the first epoch such that there exists a node $w \in \mathcal{C}_u$ that is marked during step 2 of epoch $i$ and there exists a cluster $\mathcal{C}_{u'} \neq \mathcal{C}_{u}$ and a node $w' \in \mathcal{C}_{u'}$ such that $w$ and $w'$ are neighbors. Since in each epoch each cluster can grow at most $\frac{4\alpha\log^2 n}{\log n/D}$ hops, in the trimming steps of epoch $i$ (steps 2 to 4), each cluster only backtracks by at most $\frac{4\alpha\log^2 n}{\log n/D}$ hops. But then, at the end of trimming, the internal statuses are permanent and thus, are not altered later. Hence, at the end of epoch $i$, there exist two nodes $v$ and $v'$ that are within distance $\frac{8\alpha\log^2 n}{\log n/D}$ of each other and are permanently assigned to $\mathcal{C}_u$ and $\mathcal{C}_{u'}$ as internal nodes, respectively. The factor of $2$ is because both clusters can backtrack. This shows that the connectivity gap will be at most $\frac{8\alpha\log^2 n}{\log n/D} = O(\frac{\log^2 n}{\log n/D})$.
\end{proof}

Now that we achieved a clustering with connectivity gap $O(\frac{\log^2 n}{\log n/D})$, we grow the clusters from this point onwards slowly and carefully, using an algorithm that we call Cluster-Refinement and explain next.

\paragraph{The Cluster-Refinement algorithm} We first briefly recap on the clustering obtained above which is the starting point of the Cluster-Refinement algorithm. This starting clustering has a connectivity gap of $O(\frac{\log^2 n}{\log n/D})$; each node is either clustered or unclustered and each clustered node might be internal or boundary. Furthermore, each internal node has no unclustered neighbor or a neighbor from a different cluster, and each clustered node $v$ is connected to its candidate via a path made of only internal nodes of the related cluster. 

Now we will grow the clusters slowly, and when two clusters meet, we stop their growth from that side. More precisely, we do as follows: 

As usual, we have two statuses for the clustered nodes, \emph{internal} and \emph{boundary}. The internal nodes remain internal permanently but the boundary nodes might become internal as we grow the clusters. Moreover, we have two types of boundary nodes, \emph{active} and \emph{inactive}. At the start, all boundary nodes are active. 

We divide the time into \emph{epochs}, each consisting of $\Theta(\log^2 n)$ rounds. In total, we use $O(\frac{\log^2 n}{\log n/D})$ epochs, i.e., $O(\frac{\log^4 n}{\log n/D})$ rounds. In each epoch, we grow each cluster from each side by (at most) one hop. For that, the internal nodes and the active boundary nodes run one long-phase of the basic Decay protocol, trying to transmit their cluster ID. The already clustered nodes do not listen but the unclustered nodes listen and if an unclustered node receives a message, it joins the respective cluster. Hence, if an unclustered node $w$ had a clustered neighbor $w'$ that was not an inactive boundary, $w$ becomes clustered in this epoch with high probability. This follows from \Cref{lem:long-phase}. 

After that, we are done with the cluster growing part of this epoch and we determine the new boundaries. The previously internal nodes remain internal. The general rule is that, a clustered node $v$ is boundary if it has a neighbor from a different cluster or an unclustered neighbor. In the former case, $v$ is a permanent boundary, while in the latter, it is a temporary boundary and might become internal later. If a clustered node $v$ does not become boundary, it is considered internal.

To determine these boundaries, we use a procedure similar to the step 2 of the Cluster algorithm, using $\Theta(\log^2 n)$ rounds. In these rounds, each unclustered node simply runs $\Theta(1)$ long-phases of the basic Decay protocol, sending a message declaring that it is unclustered. A clustered node that receives a message from an unclustered node becomes boundary (inactive or active depending on the part explained next). From the viewpoint of the clustered nodes, these rounds are divided into $\Theta(\log n)$ \emph{parts}, each part consisting of one phase of the basic Decay protocol. In each part, nodes of each cluster \emph{unanimously} decide to either ``try to send" or remain listening, each with probability $\frac{1}{2}$. In each part, all nodes that are trying to send perform one phase of the basic Decay algorithm, sending their cluster ID. If a clustered node receives a message of a different cluster, it becomes an inactive boundary. The same arguments as in the proof of \Cref{prop:marks} show the following:

\begin{proposition} With high probability, the boundaries are as desired. That is, a clustered node $v$ that has a clustered neighbor from a different cluster becomes inactive boundary, else if it has an unclustered neighbor it becomes boundary but active, and finally, if neither of these two are the case, it becomes an internal node of its cluster.
\end{proposition}

As the result of running the Cluster-Refinement algorithm after the cluster algorithm, we get the following:
\begin{corollary}\label{lem:cluster} Given a set of candidates, there exists a distributed algorithm in the model without collision detection that with high probability achieves a clustering around candidates with connectivity gap of at most $3$, in $O(D\log{\frac{n}{D}} + \log^3 n)$ rounds.
\end{corollary}
\begin{proof} This algorithm is simply first running the cluster algorithm explained above and then running the Cluster-Refinement algorithm. The correctness analysis are as discussed above. The round complexity of the whole procedure used above is $O(D\log{\frac{n}{D}}+ \log^2 n) + O(\frac{\log^2 n}{\log \frac{n}{D}}) \cdot O(\log^2 n)$. It is easy to see that this is in $O(D\log{\frac{n}{D}} + \log^3 n)$, considering the values that $D$ can take.
\end{proof}
%----------------------------------------------------------------------------------------------------------------------------------------------------------------------------------------------------------------------------------------------------------------------------------------------------------

\subsection{The Sparsified Overlay Graph $H'$ and the Overlay Communications}\label{sec:overlaynoCD}
Now, we use the clusters obtained in the clustering component (\Cref{lem:cluster}) to design a directed \emph{overlay} graph between the candidates. We first explain the communication actions in this overlay and then explain how to use these communications to construct the overlay.

\paragraph{Communication actions:} In this overlay, we have three types of \emph{communication actions} in or amongst the clusters. These communication actions are as follows. (\textbf{a}) {Uplink}: a candidate delivers its message to all internal nodes of its cluster. (\textbf{b}) {Intercommunication}: internal nodes exchange messages with other internal nodes of other clusters which are at distance at most $10$. (\textbf{c}) {Downlink}: internal nodes send some messages towards the candidate; we guarantee that candidate gets at least one of them. We next explain the implementation details of these communication actions.

\newparagraph{Implementation of communication actions:} In the following, we explain how we implement the communications actions Uplink, Intercommunication, and Downlink and what are the time complexities of these actions.
\begin{enumerate}
\item In uplink, candidates start with messages for transmission. Every internal node that receives a message runs a Fast-Decay($\log \frac{n}{D}$) algorithm. However, the boundary or unclustered nodes do not participate in transmissions. It is easy to see that after broadcast time $T_{BC}$, each internal node receives the message of the respective candidate. Thus, the complexity of uplink is simply $T_{BC} = O(D \log \frac{n}{D} +\log^2 n)$.  

\item In downlink, we do exactly the opposite of uplink. This time, some internal nodes start with messages and we want to deliver at least one of these messages to the related candidate. Again, every internal node (other than candidates) that has a message (or receives a message) runs a Fast-Decay($\log \frac{n}{D}$) algorithm, and the boundary or unclustered nodes do not participate in transmissions. Again, by properties of the Fast-Decay($\log \frac{n}{D}$) algorithm, after $T_{BC}$ time, the candidate receives the message of at least one internal node of the related cluster. For the sake of cleanness, the candidate can ignore all but the first received message. Clearly, the complexity of downlink is also $T_{BC} = O(D \log \frac{n}{D} +\log^2 n)$.

\item The intercommunication consists of $\Theta(\log^3 n)$ rounds which are divided into $\Theta(\log^2 n)$ epochs. In each epoch, all internal nodes of each cluster randomly decide to be active or to be listening unanimously, with probability $\frac{1}{\log n}$ for being active%
\footnote{This requires $\log^2 n$ bits of randomness shared in the cluster. We can get this shared randomness in time $O(D \log \frac{n}{D} + \log^3 n)$ by sending $k = \log n$ packets of $\log n$ random bits each via a very simple single source $k$-message broadcast algorithm from \cite{FOCS} with time complexity $O(D \log \frac{n}{D} + k \log^2 n)$.}. Then, in each epoch, internal nodes of the active clusters and all the unclustered or boundary nodes that receive a message perform $3$ phases of the basic Decay protocol based on their own local coins, where a node that receives a message in one phase will be transmitting it in the next ones. We say a cluster is globally isolated in an epoch if this cluster is the only cluster that is active in that epoch. Note that, in each epoch, each cluster $\mathcal{C}_{u}$ has a probability of $\frac{1}{\log n}$ to be isolated. Thus, using a Chernoff bound, we get that during the $\Theta(\log^2 n)$ epochs, w.h.p., there are at least $\Theta(\log n)$ epochs during which $\mathcal{C}_{u}$ is globally isolated. Then, from properties of the Decay protocol, we know that in each such epoch, each (internal) node of any other cluster that is within distance $3$ of internal nodes of $\mathcal{C}_{u}$ receive message of $\mathcal{C}_{u}$ with constant probability. Thus, after $\Theta(\log n)$ such epochs, each such node receives this message with high probability. Then, using union bounds, we get that this holds for any such node, and also in another level, for every cluster $\mathcal{C}_{u}$. This proves the correctness of the intercommunication algorithm.
\end{enumerate}

\medskip
\paragraph{Overlay:} Now using the above communications, we get a sparse \emph{overlay} $H'$ with following properties: (1) each candidate has exactly one incoming edge. If $(w, v)$ is the edge going from $w$ to $v$, we say that $w$ is the \emph{parent} of $v$, and $v$ is a \emph{child} of $w$, (2) each candidate can send a message to all its children, in one round of communication atop the overlay (3) each candidate knows all of its children, if there are at most $5$ of them, and it knows at least $5$ of them, if there are more, and can receive from these known children (at most $5$) in a constant number of rounds of communication atop the overlay, (4) each round of communication atop this overlay takes $O(D \log \frac{n}{D} +\log^3 n)$ rounds, and finally (5) this overlay is built in time $O(D \log \frac{n}{D} +\log^3 n)$. 

\newparagraph{Implementation of the Overlay atop communication actions:} Above these abstractions of communication actions, the algorithm for designing the overlay to get the aforementioned properties (1) to (5) is as follows. First, candidates send their ID to all internal nodes of their cluster using the Uplink. Then, we use intercommunication so each internal node knows the ID of clusters that are within distance $10$ of it. Then, internal nodes use the down-link to send the ID of these adjacent clusters to their respective candidates. For each candidate $u$, the first adjacent candidate that $u$ hears about becomes parent of $u$. Then, using an up-link communication, the candidates send the ID of their parent to the internal nodes of their clusters. After that, we use another intercommunication where internal nodes inform close-by internal nodes of who their parents are. Hence, after this, internal nodes of each cluster, altogether, know which clusters are their children. Now, the goal is to inform a candidate either about all its children or at least $5$ of them if it has $5$ or more. For this, we use $5$ turns of the down-link and up-link where in each turn, the candidate asks ``\emph{So far I know about children listed as follows =[...]. Tell me something new}''. With this question answered $5$ consecutive times a candidate indeed gets to know about one more child per iteration if such a child exists. In total it therefore either knows all its children learns about at least $5$ of them in the case that there are $5$ or more.% Using Uplinks and Intercommunications and then, with the exact same transmissions as done in the above Downlinks, we can communicate from each candidate to its children, and also back from all up to $5$ of its children to the candidate itself. 

%----------------------------------------------------------------------------------------------------------------------------------------------------------------------------------------------------------------------------------------------------------------------------------------------------------
\subsection{Implementing the Modified Elimination Algorithm}\label{sec:debatenoCD}

As the last component of a debate (after clustering and overlay design), we implement the Modified Elimination algorithm (MEA) on our overlay graph. Each candidate knows whether its degree is less than $5$ or not, and in the former case, the node also knows the degree exactly. Thus, each node knows its degree rounded down to $5$. Using the overlay graph, each candidate sends this \emph{flattened} degree and its ID, first, to its children, and then, to its parent. Given property (2) of the overlay, sending this message to children is straightforward. On the other hand, using property (3) of overlay, we know that each candidate receives the message of all up to $5$ of its children. After these message exchanges, each candidate uses the MEA algorithm to decide whether it remains alive or becomes removed from candidacy. If a candidate $v$ has a degree higher than $5$, it remains alive. Otherwise, noting the aforementioned properties (2) and (3), $v$ receives all the messages it needs. Thus, $v$ can compare its $(degree(.), ID(.))$ pair with the pairs in the received messages and decide about remaining alive or being removed accordingly. 

\begin{proof}[Proof of \Cref{lem:mainnoCD}] What remains to show is that the above implementation of the Modified Elimination algorithm on top of the overlay graph satisfies the desired properties of a debate, as mentioned at the start of the section. Note that if the number of remaining candidates is greater than one, the overlay graph might be disconnected. For the purpose of analysis, we look at each connected component of the overlay graph separately. If there is only one candidate remaining in the network, we are done with the proof. Otherwise, we know that each connected component has more than one candidate (assuming more than one candidates are remaining). Hence, the Modified Elimination algorithm removes at least $\frac{1}{10}$ fraction of the nodes of each connected component. This mean that it removes at least $\frac{1}{10}$ fraction of all the candidates. Also, it is clear that always at least one candidate remains alive. This is because one candidate remains alive in each connected component of the overlay graph. These two show that the two desired properties of the debates are indeed satisfied. 
\end{proof}

%\iffalse
\subsection{Reducing the Running Time of a Debate for Networks with Large Diameter}\label{sec:linearnoCD}

In this subsection we give the proofs and details for reducing the running time of our algorithm when the diameter $D$ is (almost) linear. The high-level idea for this was given in \Cref{sec:overviewlinear}. We first show that growing the cluster in debate $i$ only up to $\min\{D, \frac{4n\; 1.05^i}{\log n}\}$ instead of $D$ still leads to at least a $\frac{1}{20}$-fraction of candidates being removed in each round and therefore to a correct leader election. Then, we show that this reduces the time complexity by replacing the $\log \log n$ factor in our running time to the claimed $\log \min\{\log n,\frac{n}{D}\}$ factor. In particular, this makes our leader election algorithm run in optimal $O(n)$ rounds.

\newparagraph{Correctness Analysis:} 
We first prove the main observation claim presented in \Cref{sec:overviewlinear}. That is, we show that if the number of remaining candidates is $k$, then at most half of the candidates can be such that there is no other candidate within $2n/k$ steps. This can be easily seen by the following potential argument: Suppose each node starts with one unit of potential and then passes this potential to the candidate closest to it. Any candidate $u$ that does not have any other candidate within its $4n/k$ hops receives a potential of at least $2n/k$. Since the initial potential is $n$ the number of such candidates can be at most $k / 2$.

Now using this observation, we prove that in the new algorithm where growing the cluster in debate $i$ is done only for up to $\min\{D, \frac{4n\; 1.05^i}{\log n}\}$ radius, we still have the desired progress property of debates. That is, with high probability, for each $i$, after the $i^{th}$ debate, the number of remaining candidates would be at most $C \log n \; 1.05^{-i}$. The proof is by induction. The base case of $i=0$ is easy and as before. For the inductive step, we assume that at the start of debate $i$, the number of remaining candidates is at most $C \log n \; 1.05^{-i+1}$. We prove that after the $i^{th}$ debate, this number is at most $C \log n \; 1.05^{-i}$. If at the start of the debate, the number of remaining candidates is already less than $C \log n \; 1.05^{-i}$, then we are done. In the more interesting case, if the number of remaining candidates at the start of $i^{th}$ debate is at least $C \log n \; 1.05^{i}$, then growing the clusters for radius $\min\{D, \frac{4n\; 1.05^{-i}}{\Theta(\log n)}\}$ makes sure that at most $\frac{1}{2}$ of candidates are isolated in the overlay graph. Then, from analysis of previous sections, we know that after the debate, at least $\frac{1}{10}$ fraction of the non-isolated candidates are removed. Thus, at least $\frac{1}{20}$ of the whole remaining candidates are removed after the debate. Hence, we conclude that the number of remaining candidates after the $i^{th}$ debate is at most $(1-1/20) C \log n \; 1.05^{-i+1} < C \log n \; 1.05^{-i}$.

\newparagraph{Time Complexity Analysis:} 

Next, we study how the reduced cluster radii affect the total running time of our algorithms. For this we can restrict ourselve to the case that $D > \frac{n}{\log n}$ since we do not claim any improvement for smaller $D$. Given this summing up the radius dependent term $O(d \log\frac{n}{d})$ over all debates leads to the following:

$$\sum_{i=1}^{\Theta(\log\log n)} O(\min\{D, \frac{4n\; 1.05^{i}}{\log n}\} \cdot \log \frac{n}{\min\{D, \frac{4n\; 1.05^{i}}{\log n}\}}) \ =$$ 
$$= \sum_{i=1}^{\Theta(\log(\frac{D \log n}{n}))} O(\frac{4n\; 1.05^{i}}{\log n} \cdot \log \frac{n}{\frac{4n\; 1.05^{i}}{\log n}}) \ \; + \ \ \ \ \ $$
$$\ \ \ \ \ \ \ \ \ + \sum_{i=\Theta(\log(\frac{D \log n}{n}))}^{\Theta(\log\log n)} O(D \cdot \log \frac{n}{D})\ =$$
$$= O(D \log \frac{n}{D} + \log \frac{n}{D} \cdot D \log \frac{n}{D}) = O\left(D \log\frac{n}{D} \cdot \log \frac{n}{D}\right)$$

\iffalse
%One column formatting:
\begin{eqnarray}
&&\sum_{i=1}^{\Theta(\log\log n)} \min\{D, \frac{4n\; 1.05^{i}}{\log n}\} \cdot \;\;\log \big(\frac{n}{\min\{D, \frac{4n\; 1.05^{i}}{\log n}\}}\big) \nonumber\\ 
%
= &&\sum_{i=1}^{\Theta(\log(\frac{D \log n}{n}))} \frac{4n\; 1.05^{i}}{\log n} \cdot \log \frac{n}{\frac{4n\; 1.05^{i}}{\log n}} \ \ \ \ \ + \ \ \ \ \  \sum_{i=\Theta(\log(\frac{D \log n}{n}))}^{\Theta(\log\log n)} D \cdot \;\;\log \big(\frac{n}{D}\big)\nonumber \\
%
= &&D \log \frac{n}{D} + \log \frac{n}{D} \cdot D \log \frac{n}{D} = O\bigg(D \log\frac{n}{D} \,\cdot\, \log \frac{n}{D} \bigg) \nonumber
\end{eqnarray}
\fi
%
This shows that the total running time is at most 
$$O(D \log\frac{n}{D} \min \{\log \log n,\log \frac{n}{D}\} + \log^3 n \log \log n) = $$
$$= O(D \log\frac{n}{D} + \log^3 n) \cdot \min \{\log \log n,\log \frac{n}{D}\})$$
as claimed in \Cref{lem:mainnoCD}.

%\fi
%\pagebreak
\section{Leader Election via Beeps}
This section provides the technical details and proofs for \Cref{lem:mainCD}.  As described in \Cref{sec:overview}, this algorithm follows the template given in \Cref{sec:template} and uses the ideas explained in \Cref{sec:overviewnoCD} to implement this template. The reader is advised to read Sections \ref{sec:template} and \ref{sec:overviewnoCD} before reading this section.

In this section, we present a leader election algorithm for the beep model, which has time complexity $O(D + \log n \log \log n) \cdot \log \log n$ rounds. The outline of this algorithms is the same as the one presented in \Cref{sec:overview} as Algorithm \ref{alg:LE}. Starting with $\Theta(\log n)$ candidates the algorithm runs in $\Theta(\log \log n)$ debates. In each debate, we reduce the number of remaining candidates by a constant factor, while keeping the guarantee that at least one candidate remains. Each debate consists of a clustering phase and then an implementation of a debate using the induced overlay graph and overlay communication protocols. In this section we describe these implementations. For this we first introduce superimposed codes in \Cref{sec:codes}. Then, in \Cref{sec:debateCD1} we give a simple implementation of a debate in $O(D + \log^3 n)$ rounds and then finally in \Cref{sec:debateCD2} we show how to improve this to $O(D + \log n \log \log n)$ rounds. This running time for a debate leads to the (near) optimal leader election algorithm promised in \Cref{lem:mainnoCD}.

\subsection{Superimposed Codes}\label{sec:codes}
In this subsection, we define the two types of superimposed codes and show their existence using simple randomized algorithms. The first code is one that allows us to reconstruct messages from a collision of a bounded number of them. This code is presented in \Cref{subsubsec:fullCode} and we use it in \Cref{sec:debateCD1} it to present an implementation of debates which is closer to the outline presented in \Cref{alg:debateTemplate}. However, we use this code only for explaining the general idea of the approach. We present another family of codes in \Cref{subsubsec:countCodes} which provide only an approximate-count, and are thus simpler and shorter. Our most efficient algorithm, which we present in \Cref{sec:debateCD2}, only uses this second type. 
 
However, we note that similar to most of the other versions, our version is also constructed using very simple randomized algorithms.
\subsubsection{The Message Exchange Codes}
\label{subsubsec:fullCode}
\begin{definition}
A $k$-superimposed code or $SI(k)$-code of length $l$ for a finite set $N$ assigns each element in $N$ a binary codeword of length $l$ 
such that (1) every superimposition of $k$ or less codewords is unique and (2) every superimposition of more than $k$ codewords is 
different from any superimposition of $k$ or less codewords.
\end{definition}

It is easy to see that good superimposed codes of short length exist. We next present a simple randomized construction for this. We note there are many known versions of the superimposed codes and the general notion of the codes we have is somewhat similar to those presented in \cite{indyk1997SIcodes}. Also, the proof style we present here is the typical approach. The code we prove in \Cref{lem:code} can be also 
generated by using those of \cite{SI-Codes-Survey}\footnote{We also note that a less efficient but deterministic version, with size $\Theta(k^2 \log^2 N/\log^2\log k)$ bits, follows from \cite{Kautz-Singleton}.}. Since the proofs are simple, we provide the direct and self-contained proof here. We also note that, we do not use the codes provided by \Cref{lem:code} in our most efficient algorithm---they will be replaced by the approximate counting codes of the next subsection. We however use them for a better and simpler explanation of the general approach.

\begin{lemma}\label{lem:code}
For every $N$ and any $k$ there exists a $SI(k)$-code for $N$ of length $l = 4(k+1)^2 \log N$.
\end{lemma}
\begin{proof}
We show that a random code $C$ in which each codeword position is set to one with probability $p = 1/(k+1)$ has the
desired properties with good probability. 
To see this we take any $k+1$ codewords $c_0,c_1,\ldots,c_{k}$ and note that the probability that there is no
position for which $c_0$ is one and all codewords $c_1,\ldots,c_{k}$ are zero is exactly $(1 - p(1 - p)^{k})^l$ which is at most
$$(1 - \frac{1}{e(k+1)})^l < (1/e)^{(k+1)\log N} < \frac{N^{-(k+1)}}{2}.$$ 
Taking a union bound over all $\binom{N}{k+1} < N^{k+1}$ choices of codewords we get that with probability at least $1/2$, we 
get the property that the superimposition of any $k$ codewords differs from any different codeword $c_0$ by having a zero 
where $c_0$ has a one. In particular, this implies that given two sets $S$ and $S'$ with $S' \geq |S|$, $|S| \leq k$ and $S \neq S'$
the superimposition of all codewords in $S'$ has a one on a position in which the superimposition of all codewords in $S$ does not.
This is true because $S'$ contains at least one codeword $c_0$ that is not in a superset of $S$ of size $k$ and making $S$ smaller
and $S'$ larger does not change this fact. It is easy to see that both property (1) and (2) now follow directly.
\end{proof}

\subsubsection{The Approximate Counting Codes}
\label{subsubsec:countCodes}
We also use the following approximate counting superimposed code:

\begin{definition}
For any $N$, $k < N$ and $\delta>0$, a $(1+\delta)$-approximate $k$-counting superimposed code of length $l$ consists of a distribution $\mathcal{D}$ over binary codewords of 
length $l$ and a decoding function $decode: \{0,1\}^l \rightarrow [k]$ such that for every $j \in [k]$ and codewords $c_1,\ldots,c_j$ independently sampled from $\mathcal{D}$
we get that:
$$\Pr\left[\ \frac{j}{1+\delta} \ \leq \ decode\left(\bigoplus_{i=1}^j c_i\right) \ \leq \ j(1+\delta) \ \right] \geq 1 - 1/N.$$
\end{definition}

As the next lemma shows, there are such code with length that has only logarithmic dependence on $k$ and $N$ and polynomial dependence on $1/\delta$. In our applications, we use constant $\delta$.

\begin{lemma}\label{lem:approxcodesize}
For any any $N$, $k < N$ and $\delta>0$, there exists a $(1+\delta)$-approximate $k$-counting superimposed code of length $l = \Theta(\frac{\log N \; \cdot\; \log k}{\delta^3})$.
\end{lemma}
\begin{proof}
We first explain the encoding. Each codeword in the distribution $\mathcal{D}$ of the codes we construct consists of $\log_{1+\delta} k=\Theta(\log k /\delta)$ blocks, each having $\Theta(\log N/\delta^2)$ bits. In the $i^{th}$ block, each bit of each codeword is independently set to one with probability $p_i=1-2^{-1/(1+\delta)^i}$ and to zero otherwise. 

We now explain the decoding. The decoding function receives a binary word of length $l$ and dissects it into its blocks. Then in this word, the decoding finds the \emph{largest} $i^*$ such that the majority of the bits in the $i^*$ block are equal to one. Then, the decoding outputs $(1+\delta)^{i^*}$ as the estimate.

Finally, we present the correctness analysis. To show that this works, consider codewords $c_1,\ldots,c_j$ independently sampled from $D$ and let $C=\bigoplus_{i=1}^j c_i$, and let $i^{*}$ be the largest value such that the majority of the bits in the block $i^*$ of $C$ are equal to one. We show that $$\Pr\left[\ \frac{j}{1+\delta} \ \leq \ (1+\delta)^{i^*} \ \leq \ j(1+\delta) \ \right] \geq 1 - 1/N.$$

Consider an arbitrary block $i$. First, we show that if $(1+\delta)^i\geq j(1+\delta)$, then with probability $1-\frac{1}{N^3}$, the majority of the bits in the $i^{th}$ block are zero. For each bit, the probability that it is one is at most $1-(1-p_i)^j = 1-2^{-j/(1+\delta)^i} \leq 1-2^{-1/(1+\delta)} \leq 0.5 - \frac{\delta}{4}$. Thus, an application of the Hoeffding bound shows that, the probability that in $\Theta(\log N/\delta^2)$ bits, the majority are one is at most $1/N^3$. We now show that if $(1+\delta)^i\leq j/(1+\delta)$, then with probability $1-\frac{1}{N^3}$, the majority of the bits in the $i^{th}$ block are one. For each bit, the probability that it is one is at least $1-(1-p_i)^j = 1-2^{-j/(1+\delta)^i} \geq 1-2^{-(1+\delta)} \leq 0.5 + \frac{\delta}{4}$. Again, it follows from the Hoeffding bound that the probability that in $\Theta(\log N/\delta^2)$ bits, majority are zero is at most $1/N^3$. 

Now recall that $i^{*}$ is the largest value such that the majority of the bits in block $i^*$ are equal to one. Having analyzed the two sides, using a union bound over all $k<N$ values for $j$ and all blocks, we get that with probability at least $1-\frac{2}{N^2}\geq \frac{1}{N}$, we have $\frac{j}{1+\delta} \ \leq \ (1+\delta)^{i^*} \ \leq \ j(1+\delta)$.  

%we note that in expectation exactly a $1 - (1 - 2^{-i})^{j}$ fraction of block $b_i$ is ones and for a large enough block length $\Theta(\log N)$ a Chernoff bound shows that the probability that this fraction deviates by a $0.9$ factor or more is at most $2/N^3$. A union bound over all $k < N$ values for $j$ and all $\log k < N$ values for $i$ then shows that the probability of a too small estimation is at most $2/N$. The analog argument also shows that an overestimation happens at most with probability $2/N$ which completes the proof. 
\end{proof}

\subsection{$O(D + \log^3 n)$-length debate}\label{sec:debateCD1}
The outline of this debate algorithm is exactly that of the simple debate algorithm sketched in \Cref{sec:template}. We first grow clusters around candidates, then each candidate find its degree in the overlay graph, then candidates exchange their ($degree$, $ID$) pairs, and at the end, each candidate remains a candidate only if its pair is greater than all the pairs that it received. Next, we zoom in on how we implement each of these steps with beeps.

\begin{algorithm}[t]
\caption{Debate-1 Algorithm, run @ node $u$}
\begin{algorithmic}[1]
\algorithmsize
\Statex
\State Cluster \Comment{Step 1}
\Statex
\State Uplink coded IDs \Comment{Step 2}
\State Intercommunicate
\State Downlink
\Statex
\If {$candidate$}
	\State $S_u\gets$ decoding of received message as set of IDs
	\State $s_u \gets |S_u|$
\EndIf
\Statex \Comment{Step 3}
\State Uplink $C(s_u, ID_u)$
\State Intercommunicate 
\State Downlink
\Statex
\If {$candidate$} \Comment{Step 4}
	\State $T \gets$ decoding of received messages as set of ordered pairs of degree and ID
	\If {received a pair greater than that of $u$} 
		\State $candidate \gets false$
	\EndIf
\EndIf
\end{algorithmic}
\label{alg:debate}
\end{algorithm}

\newparagraph{Clustering:}For clustering nodes via beeps, we assign each node to the cluster of the candidate which is closest to it. In the case of a tie, that is, when there is more than one closest candidate, we leave the node  unclustered. 

To achieve this clustering goal, we use an Uplink. Trying to adapt to the superimposition nature of the beeping model, we re-define Uplink action as follows. For each node $u$, we denote by $dist(u)$ the distance of $u$ to the closets candidate. In Uplink, each candidate has a message of length $L$ for transmission, and we want each node $u$ to receive the superimposition of the messages of the candidates at distance $dist(u)$ from $u$. We later see why this Uplink procedure is a natural fit to the beeping model and how we can implement this Uplink easily. Before going to these implementation related details, we finish the discussion about clustering. Suppose that there exists a Black-box algorithm $\mathcal{A}_{up}$ for the above Uplink description. Now we explain how to use $\mathcal{A}_{up}$ to cluster nodes in the desired manner. For this, we use a SI($1$) code. That is, each candidate encodes its ID using this code, and candidates Uplink these coded IDs. On the receiving end, each node $u$ receives the superimposition of the coded IDs of the candidates at distance $dist(u)$ from $u$. Noting the properties of SI($1$) codes, if there is only one such candidate, $u$ can decode the ID of that candidate. On the other hand, if there are two or more of those candidates, $u$ can distinguish this case and declare itself as \emph{unclustered}. This concludes the clustering task. 

One remark about the shape of the clusters achieved by this algorithm is as follows: Each node $w$ can be unclustered only if it has two neighbors $v_1$ and $v_2$ that belong to different clusters. Thus, each cluster grows from every side till it either reaches the margins of the network or it is within distance $2$ hops from another cluster. We call a node $u$ \emph{boundary} if $u$ is clustered but it is adjacent to a node $u'$ such that $u'$ is either unclustered or it belongs to a cluster other than that of $u$. We say two clusters $\mathcal{C}_1$ and $\mathcal{C}_2$ are \emph{adjacent} if there exist two nodes $v_1 \in \mathcal{C}_1$ and $v_2 \in \mathcal{C}_2$ such that $v_1$ and $v_2$ are within distance $2$ of each other. It is clear that in that case, $v_1$ and $v_2$ are boundary nodes. If the distance between $v_1$ and $v_2$ is exactly $1$, then clusters are directly touching each other, whereas if distance is two, with an unclustered node $w$ in the middle, then $w$ serves as a bridge connecting the two clusters.

\begin{algorithm}[t]
\caption{\small{Clustering, run @ node $u$}}
\begin{algorithmic}[1]
\algorithmsize
\Statex
\State Numbering
\State $C \gets SI(1)$-code
\State $m_u \gets C(0, ID_u)$
\State Uplink $m_u$, receive bit-sequence $m'_{u}$
\Statex
\If {$m'$ is a valid ID}
	\State $Cluster$-$ID \gets$ decoding of $m'$ into an ID
	\State $clustered \gets true$ 
\Else 
	\State $Cluster$-$ID \gets \emptyset$
	\State $clustered \gets false$
\EndIf
\Statex
\State $boundary \gets false$
\For {t=0 to $L-1$}
	   \If {$m'[t]=1$}
			\State \textsc{beep}
			\State \textsc{listen}
		\Else
			\State \textsc{listen}
			\State \textsc{beep}
		\EndIf
		\If {heard a beep while listening}
			\State $boundary \gets true$
		\EndIf
\EndFor
\end{algorithmic}
\label{alg:clustering}
\end{algorithm}

\newparagraph{Communications on the Overlay Graph:} For implementing communications, we want to devise protocols such that using these protocols, each candidate can exchange messages with neighboring candidates in the overlay graph. Trying to adapt to the superimposition nature of beeping networks, we do this in two layers: we first implement communications between candidates such that each candidate receives the superimposition of the messages of the neighboring candidates, in the overlay graph. Then, we use a $SI(\log n)$ code on top of these superimposition channels to get to full message exchange. Note that $SI(\log n)$ codes are robust enough because the number of candidates is at most $\log n$. On the negative side, these codings come with a cost, the encoding of the $\Theta(\log n)$ bit messages is $\Theta(\log^3 n)$ bits, which leads to the $\log^3 n$ term in the time bound of the debates. later we explain how to modify the debate algorithms to get over this cost.

Thus, what remains is to implement communications between candidates such that each candidate receives the superimposition of the messages of the neighboring candidates, in the overlay graph. For this, we first \emph{number} each node $u$ with its distance from the closest candidate $dist(u)$. This \emph{numbering} is essentially the backbone of the clusters and serves as the spine of our intra-cluster communications.

\paragraph{Numbering:} The algorithm for numbering is simple and some pseudo-code of it is given as Algorithm \ref{alg:numbering}. In each round, each node is active or inactive; at the start, only candidates are active; and each node simply records the time in which it becomes active. In each round, active nodes beep and each inactive node becomes active if it hears a beep. This way, the wave of the activation (the wave of beeps) proceeds exactly one hop in every round. The reader might find \Cref{fig:Onewave} in understanding the concept of a beep wave. We get that every node $u$ gets activated after exactly $dist(u)$ rounds where $dist(u)$ is the distance of $u$ to the closest candidate. 

\begin{algorithm}[t]
\caption{\small{Numbering Algorithm, run @ node $u$}}
\begin{algorithmic}[1]
\algorithmsize
\Statex
\Statex \textbf{Output}: distance $dist(u)$ to the closest candidate
\Statex
\State $active \gets false$
\If {$candidate$}
	\State $ dist(u)\gets0$
\EndIf
\For {$t=1$ to $D$}
	\If {$active$ or $candidate$}
		\State \textsc{beep}
   \Else
		\State \textsc{listen}
		\If {heard a beep}
			\State $active \gets true$
			\State $dist(u) \gets t$
		\EndIf
	\EndIf
\EndFor
\end{algorithmic}
\label{alg:numbering}
\end{algorithm}

Having this numbering, we implement the communications between candidates via three communication actions in or amongst the clusters: Uplink, Intercommunication, and Downlink. However, we change the definitions of these three tasks to adapt them to the superimposition nature of beeping model. 
\begin{itemize}
\item[(a)] In the default definition of uplink, each candidate starts with one message and uplink delivers the message of each candidate to the nodes in its cluster. In the adapted definition, we deliver to each node $u$, the superimposition of messages of candidates that are at distance $dist(u)$ from $u$. Thus, in particular, each clustered node receives the message of its related candidate. Moreover, unclustered nodes receive superimposition of more than one messages. This later helps us to distinguish the clustered nodes from the unclustered ones. After the clustering, we essentially use the Uplink only for delivering the message of each candidate to the boundary nodes of its cluster. 

\item[(b)] In the default definition, intercommunication is the action where boundary nodes of different clusters exchange messages with each other. Adapting to the superimposition nature of beeping model, in intercommunication, the goal is for each boundary node to receive the \emph{superimposition} of the messages of adjacent boundary nodes. 

\item[(c)] In the default definition, downlink is where the message is brought down from the boundary nodes to the candidates. Adapting to the beep model, the new goal is for every candidate to receive the superimposition of the messages that boundaries of its cluster send.
\end{itemize}

Having these new definitions, now we present the implementation details of these tasks, and see why these new definitions are easy to implement in the beeping model, thanks to the numbering that we have created. 
\bigskip

\noindent\underline{Uplink:} The algorithm is as presented in Algorithm \ref{alg:uplink}. The main technique in here is the usual idea of pipelining the beeps. First, consider what happens inside one cluster ignoring the effect from the other clusters. A node at distance $dist(u)$ does its transmission about the $\ell^{th}$ bit of the message at round $dist(u)+3l$. In particular, the candidate starts the transmission about the first bit in round $0$ and it finishes its transmissions in round $3L$. In each round $t$, a node $u$ is allowed to transmit a bit only if $t - dist(u) \equiv 0 \pmod{3}$. In that round, nodes that are one hop away are listening to this bit. That is, node $w$ is listening to this bit (and recording it) if $t - dist(u) \equiv 2 \pmod{3}$. We first consider what happens to the first bit of the message. In the first round, the candidate transmits or remains silent depending on what is the first bit of message. Then, inductively we see that for each $i\in [D]$, in the $i^{th}$ round nodes that are at distance $i$ from source transmit or remain silent depending on the first bit of the message. This way, the first bit reaches $D$-hops away after $D$ rounds. Now note that when the first bit has traveled only three hops from the candidate, the candidate starts transmitting the second bit, and thus, the wave of the transmissions of the $(j+1)^{th}$ bit follows the wave of the transmissions of the $j^{th}$ bit with a three hops lag. Hence, by $D + 3L$ rounds, all the bits have reached every node. 

Next we explain the effect of beep waves of different clusters have on each other. Consider two neighboring clusters $\mathcal{C}_1$ and $\mathcal{C}_2$ respectively related to candidates $u_1$ and $u_2$. First suppose that $\mathcal{C}_1$ and $\mathcal{C}_2$ are connected via a bridging unclustered node $w$, where $w$ is connected to $v_1 \in \mathcal{C}_1$ and $v_2 \in \mathcal{C}_2$. Then, since different clusters grew at the same speed, we have $dist(v_1) = dist(v_2) = dist(w)-1$. Thus, using the above beep waves, $w$ always listens to the transmissions of $v_1$ and $v_2$ (and gets the superimposition of them) while $v_1$ and $v_2$ ignore the transmissions of $w$. Hence, the unclustered nodes receive the superimposition of the messages of their respective closest candidates. More importantly, the beep waves \emph{clash} at the bridging node and do not go inside the other clusters. Hence, the progresses of the beep waves inside clusters remain intact. A similar thing happens when $\mathcal{C}_1$ and $\mathcal{C}_2$ are directly touching each other. In that case, for related boundary nodes $v_1$ and $v_2$, we have $dist(v_1) = dist(v_2)$ and thus, $v_1$ and $v_2$ do not listen to the transmissions of each other.

\begin{algorithm}[t]
\caption{\small{Uplink Algorithm, run @ node $u$}}
\begin{algorithmic}[1]
\algorithmsize
\Statex
\Statex \textbf{Given}: $dist(u)$, and message bit sequence $m_u$ (for any candidate $u$)
\Statex \textbf{Output}: bit sequence $m'_{u}$ at each node
\Statex
\State $active \gets false$
\For {$t=0$ to $D + 3L-3$} 
	\If {$candidate$}
	 	\State $active \gets (m_u[\floor{t/3}] == 1)$
	\EndIf
	\Switch{ $t-dist(u)\pmod{3}$}
		\Case{$0$:}	
			\If {$active$} 
				\State \textsc{beep }
			\Else
				\State \textsc{listen}	
			\EndIf
		\EndCase
		\Case{$1$:}
			\State \textsc{listen}
		\EndCase
		\Case{$2$:} 
			\State \textsc{listen}
			\If{ heard a beep}  
				\State $m'_{u}[\floor{(t-dist(u)+1)/3}] \gets 1$
				\State $active \gets true$ 
			\Else 
				\State $m'_{u}[\floor{(t-dist(u)+1)/3}] \gets 0$ 
				\State $active \gets false$ 
			\EndIf
		\EndCase
	\EndSwitch

\EndFor
\end{algorithmic}
\label{alg:uplink}
\end{algorithm}

\bigskip
\noindent\underline{Intercommunication:} With the new definition, the intercommunication task is now easy to implement. An ideal algorithm would be like this: boundary nodes go through the bits of the messages that they have, bit by bit, and for each bit, they beep if the bit is a one, and listen otherwise. Each node record a $1$ if it beeps itself or if hears a beep. This way, if two clusters are touching, then on the related boundary nodes, the beep of one would be immediately observable by the other. However, if two clusters are connected via an unclustered bridging node $w$, then the beeps of two clusters do not reach each other. To remedy this, we do a slight modification to the above simple ideal algorithm: now for each bit, we use two rounds instead of one round. Each boundary beeps twice or listens twice depending on the bit that it has. Also, unclustered nodes listen in the first round and propagate whatever they received in the first round (beep iff they received a beep). Then, for each bit, each boundary records a one if it beeps itself or it senses a beep in any of the related two rounds. This protocol is presented in Algorithm \ref{alg:intercommunication}. It is easy to see that, this protocol achieves the desired superimposed-type intercommunication goal.

\begin{algorithm}[t]
\caption{\small{Intercommunication, run @ node $u$}}
\begin{algorithmic}[1]
\algorithmsize
\Statex \textbf{Given}: clustering, and message bit sequence $m''_u$ if $u$ is boundary
\Statex \textbf{Output}: superposition bit sequence $\mu_{u}$ if $u$ is candidate 
\Statex
\For {t=0 to $L-1$}
	\If {$clustered \;\&\; boundary$}
		\If {$m_u[t] == 1$}
			\State \textsc{beep}
			\State \textsc{beep}
			\State $m'''_{u}[t] \gets 1$
		\Else
			\State \textsc{listen}
			\State \textsc{listen}
			\If {heard a beep in above two rounds}
				\State $m'''_{u}[t] \gets 1$
			\Else
				\State $m'''_{u}[t] \gets 0$
			\EndIf
		\EndIf
	\Else
		\State \textsc{listen}
		\If {heard a beep}
			\State \textsc{beep}
		\Else 
			\State \textsc{listen}
		\EndIf
	\EndIf
\EndFor
\end{algorithmic}
\label{alg:intercommunication}
\end{algorithm}

\bigskip
\noindent\underline{Downlink:} As presented in Algorithm \ref{alg:downlink} the implementation of downlink is simply reversing the direction of beep waves of the uplink. Now, the transmissions start at the nodes furtherest away from the candidate, move towards the candidate. Nodes go through the bits with a lag of three hops between the waves related to two consequent bits. Using the transmission schedules based on the numbering, each node $v$ only listens to transmissions of nodes that are at distance $dist(v)+1$ from the candidate. In this case, $v$ receives the superimposition of messages of those nodes. Since superimposition of superimpositions is simply a superimposition, what at the end the candidate receives is the superimposition of the messages sent out from the boundary nodes.  

\begin{algorithm}[t]
\caption{\small{Downlink, run @ node $u$}}
\begin{algorithmic}[1]
\algorithmsize
\Statex \textbf{Given}: clustering, and bit sequence $\mu_u$ if $u$ is boundary 
\Statex \textbf{Output}: a bit sequence $\mu'_{u}$ in each candidate
\Statex
\State $active \gets 0$
\For {$t=D + 3L-3$ downto $0$}
	\If {$clustered \; \& \; boundary$}
	 	\If {$t-dist(u) \in [0, 3(L-1)]$}
	 		\State $active \gets \mu_u[\lfloor t-dist(u))/3\rfloor]$
	 	\EndIf
	\EndIf
	\Switch{ $t-dist(u)\pmod{3}$}
		\Case{$0$:}
			\If {$active == 1$} 
				\State \textsc{beep }
			\Else
				\State \textsc{listen}	
			\EndIf
		\EndCase
		\Case{$1$:}
			\State \textsc{listen}
			\If{ heard a beep}  
				\State $active \gets 1$ 
			\Else 
				\State $active \gets 0$ 
			\EndIf
			\If {$candidate$}
				\If {$t \in [1, 3(L-1)+1]$}
					\State $\mu'_{u}[(t-dist(u)-1)/3] \gets active$
				\vspace{0.05cm}
				\EndIf
			\EndIf
		\EndCase
		\Case{$2$:}
			\State \textsc{listen}
		\EndCase
	\EndSwitch
\EndFor

\end{algorithmic}
\label{alg:downlink}
\end{algorithm}

\subsection{$O(D + \log n \log \log n)$-length debates}\label{sec:debateCD2}

Now we show how to modify the debate algorithm presented above to get its time complexity to $O(D + \log n \log \log n)$, which leads to optimal $O(D + \log n \log \log n) \cdot \log \log n$ leader election (optimal up to $\log \log n$ factors).

As explained in the overview section, the main change is based on the following simple observation: In the debate algorithm, we do not need a full message communication and instead, only learning the following two items would be sufficient: (1) an approximation of the number of different messages received---i.e., those involved in the superimposition, and (2) whether a neighbor has a message numerically larger or not. The former knowledge is used for determining the boundary nodes during the clustering and also estimating the degree of the candidates in the overlay graph. The latter is used for detecting whether a neighboring candidate has a greater($degree$, $ID$) pair. 

In the following, we explain how to achieve these two goals without going through the high cost of full message communication. Having the implementation of these two items, the analysis of the main elimination algorithm is as presented in \Cref{sec:overviewCD} where we showed that each new debate reduces the number of remaining candidates by a constant factor while keeping at least one.

For the first purpose, instead of $SI(\log n)$ codes, we use the approximate counting codes that we presented in \Cref{subsubsec:countCodes} with $\delta=0.1$. These codes are just strong enough to enable us to approximate the number of the codewords in the superimposition to within a $1+\delta=\frac{11}{10}$ factor. These codes, encode each message of length $\Theta(\log n)$ bits into a codeword of  $\Theta(\log n \log \log n)$ bits. We use these codes to detect the boundary nodes; a $\frac{11}{10}$-approximation is enough here because the boundary nodes will receive superimposition of codewords of two or more candidates, while the (internal) nodes in the clusters receive just one codeword, that of their candidate. These codes also allow us to approximate the degree of each candidate in the overlay graph to within a constant factor, say $\frac{11}{10}$, which is sufficient for the elimination algorithm.

For the second purpose, that is, for detecting whether a neighboring candidate has a numerically larger message or not, we need a slight modification in the intercommunication algorithm. We say that boundary node $u$ should be \emph{marked} if $u$ has a node $w$ (from a different cluster) within its two hops such that the message of $w$ is numerically larger than that of $u$. In the new intercommunications, the goal is for each boundary node to detect whether it should be marked. Once the marking procedure is done, we simply use a Downlink with single-bit messages (marked or not) and each candidate gets to know whether any of nodes in its cluster is marked. This means that, each candidate knows if it has a neighboring candidate in the overlay graph with a numerically larger message or not.

For marking the boundary nodes according to above rule, the ideal algorithm is for the boundary nodes to go through the bits of their messages and compare them one by one. In each round, each unmarked boundary node beeps if the related bit of its message is one, and listens otherwise. Then, each unmarked boundary gets \emph{marked} if it was listening but heard a beep. A boundary node that gets marked does not continue the intercommunication procedure. Similar to intercommunication in previous debate algorithm, to remedy the issue that neighboring clusters might be not directly touching, we use an extra beeping round. For each bit we spend two rounds: each unmarked boundary with a $1$ in the respective bit of its message beeps twice; each other unmarked boundary listens twice and each unclustered node listens first and then repeats what it hears in the next round. An unmarked boundary gets marked if it was not beeping but heard a beep in any of the rounds. The related pseudocode is presented in Algorithm \ref{alg:intercommunication2}.

\begin{algorithm}[t]
\caption{\small{Max-Detection-Intercommunication, run @ node $u$}}
\begin{algorithmic}[1]
\algorithmsize
\Statex \textbf{Given}: clustering, and message $m_u$ if $u$ is a boundary
\Statex \textbf{Output}: boolean $marked$, for any boundary $u$
\Statex
\State $marked \gets false$
\For {t=0 to $L-1$}
	\If {$clustered \;\&\; boundary \; \& \; \neg marked$}
		\If {$bit(m_u,t)=1$}
			\State \textsc{beep}
			\State \textsc{beep}
			\State $rec_{u}[t] \gets 1$
		\Else
			\State \textsc{listen}
			\State \textsc{listen}
			\If {heard a beep in above two rounds}
				\State $marked \gets true$
			\EndIf
		\EndIf	
	\Else
		\State \textsc{listen}
		\If {heard a beep}
			\State \textsc{beep}
		\Else 
			\State \textsc{listen}
		\EndIf
	\EndIf
\EndFor
\end{algorithmic}
\label{alg:intercommunication2}
\end{algorithm}

%\iffalse
\subsection{Reducing the Running Time of a Debate for Networks with Large Diameter}\label{subsec:linearLEWCD}
In the previous section, we described an algorithm for a debate that takes $O(D+\log n \cdot \log\log n)$ rounds. Hence, this leads to a total round complexity of $O((D+\log n \cdot \log\log n)\log\log n)$ for the whole leader election algorithm. To complete the proof of \Cref{lem:mainCD}, we need to show that one can replace the multiplicative $O(\log \log n)$ factor with an $O(\min \{\log \log n,\log \frac{n}{D}\})$ factor.
Here, we present a simple change that achieves this. In particular, this provides an algorithm with round complexity of $O(n)$. 

As explained in  \Cref{sec:overviewlinear}, the idea is to grow the clusters in the $i^{th}$ debate, instead of growing the clusters to radius $D$, we grow them only up to radius of $d_i=\min\{D, \frac{4n\; 1.05^i}{\log n}\}$. We argue that this change replaces the $\log \log n$ factor in our running time by the claimed $\log (\min\{\log n,\frac{n}{D}\})$ factor. %This is most interesting when $D\geq $ is almost linear in $n$, in which case, the change makes our leader election algorithm run in optimal $O(n)$ rounds. 
The correctness analysis is exactly as presented in \Cref{subsec:linearLEWCD}. 

For the complexity analysis, we can focus our attention to the case where $D\geq \frac{n}{\log n}$. This is because, when $D< \frac{n}{\log n}$, we have $\frac{n}{D} \geq \log n$ and thus, the time complexity of $O((D+\log n \log\log n)\log \log n)$ that we already have is sufficient for proving \Cref{lem:mainCD}. To complete the proof of \Cref{lem:mainCD},  we just need to show a round complexity of $O(D \log\frac{n}{D})$ for the case For the case $D\geq \frac{n}{\log n}$. 

Recall that in each debate, growing clusters up to radius $d$ takes $O(d+\log n \log\log n)$ rounds. Since in the $i^{th}$ debate we grow the clusters up to radius $d_i=\min\{D, \frac{4n\; 1.05^i}{\log n}\}$, the new time complexities is as follows:

$$\sum_{i=1}^{\Theta(\log\log n)} O(\min\{D, \frac{4n\; 1.05^{i}}{\log n}\} +\log n \cdot \log\log n)  $$
$$= \sum_{i=1}^{\Theta(\log(\frac{D \log n}{n}))} O(\frac{4n\; 1.05^{i}}{\log n}) + \sum_{i=\Theta(\log(\frac{D \log n}{n}))}^{\Theta(\log\log n)} O(D) + O(\log n \cdot \log^2\log n)$$
$$\stackrel{(*)}{=} O(D) + O(D \cdot \log \frac{n}{D}) = O(D \log\frac{n}{D})$$
Here, the equation $(*)$ holds because the first summation in the left hand side is a geometric series and the second summation has only $O(\log \frac{n}{D})$ terms. 
%
%\iffalse
%%Full page formatting:
%\begin{eqnarray}
%\sum_{i=1}^{\Theta(\log\log n)} \min\{D, \frac{4n\; 1.05^{i}}{\log n}\} = &&\sum_{i=1}^{\Theta(\log(\frac{D \log n}{n}))} \frac{4n\; 1.05^{i}}{\log n} \ \ \ \ \ + \ \ \ \ \  \sum_{i=\Theta(\log(\frac{D \log n}{n}))}^{\Theta(\log\log n)} D \nonumber \\
%%
%= &&D + \log \frac{n}{D} \cdot D = O(D \log\frac{n}{D}) = O(n) \nonumber
%\end{eqnarray}
%\fi
%
%
%$$O(D \min \{\log \log n,\log \frac{n}{D}\} + (\log n \log \log n) \log \log n) $$
%$$= O((D + \log n \log \log n) \cdot \min \{\log \log n,\log \frac{n}{D}\}).$$

%\fi

\section{Conclusion}

In this paper we presented the first linear time distributed algorithm for electing a leader in a radio network without collision detection. 
More importantly our algorithm runs with high probability in $$O\left(D \log \frac{n}{D} + \log^3 n\right) \cdot \min\left\{\log \log n, \log \frac{n}{D}\right\}$$ rounds which is almost optimal given the $T_{BC} = \Omega(D \log \frac{n}{D})$ and $T_{BC} = \Omega(\log^2 n)$ lower bounds from \cite{KM} and \cite{ABLP} for the broadcast problem. A remaining question is to reduce the additive $\log^3 n$ term in our round complexity to the optimal $\log^2 n$. 

We also give an almost optimal $O\left(D + \log n \log \log n \right) \cdot \min\left\{\log \log n, \log \frac{n}{D}\right\}$ leader election algorithm for radio networks with collision detection and the more restricted beep networks. 

In regards to both of our algorithms, it is interesting to see if one can remove the multiplicative factor of $\log \log n$. Possibly the ideas that reduce the running time from $O(n \log \log n)$ to $O(n)$ described in \Cref{sec:overviewlinear} can be useful towards that goal. 
%Presenting a leader election algorithm that works in essentially $T_{BC}$ rounds improves over the $23$ year old simulation approach of Bar-Yehuda, Goldreich and Itai. 

% and answers their question ``whether the [$\ldots$] algorithm for leader election in arbitrary radio networks can be improved''\cite{BGI2} beyond  simulating leader election protocols in single-hop networks step by step using a broadcast for each step. 

%This improves over the deterministic algorithm from \cite{KP09} which takes $\Theta(n)$ rounds independently of the network diameter $D$.

As mentioned before, leader election is a crucial first step in communication primitives such as multiple broadcasts, multiple unicasts or message aggregation. Thus, the $\Theta(T_{BC} \log n)$ running time of leader election had became a bottleneck for getting better algorithms for these tasks. The leader election algorithm provided in this paper essentially removes this barrier barrier and now it is possible to get algorithms for these tasks that run in (almost) broadcast time $T_{BC}$. This opened the road for the results in \cite{SODA-COMM}, where the authors develop near optimal algorithms for the aforementioned communication primitives. It is interesting to see what other problems can be solved faster using this new leader election algorithm. 
}

%-----------------------------------------------------------------------------------
\bibliographystyle{acm}
\bibliography{Bdata}

\begin{thebibliography}{10}

\bibitem{AABCHK13}
{\sc Afek, Y., Alon, N., Bar-Joseph, Z., Cornejo, A., Haeupler, B., and Kuhn,
  F.}
\newblock Beeping a maximal independent set.
\newblock {\em Distributed Computing 26}, 4 (2013), 195--208.

\bibitem{ABLP}
{\sc Alon, N., Bar-Noy, A., Linial, N., and Peleg, D.}
\newblock A lower bound for radio broadcast.
\newblock {\em Journal of Computer and System Sciences 43}, 2 (1991), 290--298.

\bibitem{aghk2014}
{\sc Alon, N., Ghaffari, M., Haeupler, B., and Khabbazian, M.}
\newblock Broadcast throughput in radio networks: Routing vs. network coding.
\newblock {\em Proceedings of the Twenty-Fifth Annual ACM-SIAM Symposium on
  Discrete Algorithms\/} (2014), 1831--1843.

\bibitem{TsyMik78}
{\sc B.~S.~Tsybakov, V.~A.~M.}
\newblock Free synchronous packet access in a~broadcast channel with feedback.
\newblock {\em Problems Inform. Transmission 14}, 4 (1978), 259--280.

\bibitem{BGI1}
{\sc Bar-Yehuda, R., Goldreich, O., and Itai, A.}
\newblock On the time-complexity of broadcast in radio networks: an exponential
  gap between determinism randomization.
\newblock {\em Proceedings of the ACM Symposium on Principles of Distributed
  Computing\/} (1987), 98--108.

\bibitem{BGI2}
{\sc Bar-Yehuda, R., Goldreich, O., and Itai, A.}
\newblock Efficient emulation of single-hop radio network with collision
  detection on multi-hop radio network with no collision detection.
\newblock In {\em Proceedings of the 3rd International Workshop on Distributed
  Algorithms\/} (1989), pp.~24--32.

\bibitem{BII93}
{\sc Bar-Yehuda, R., Israeli, A., and Itai, A.}
\newblock Multiple communication in multi-hop radio networks.
\newblock {\em SIAM Journal on Computing 22}, 4 (1993), 875--887.

\bibitem{Cadambe:2008}
{\sc Cadambe, V.~R., and Jafar, S.~A.}
\newblock Interference alignment and degrees of freedom of the-user
  interference channel.
\newblock {\em Information Theory, IEEE Transactions on 54}, 8 (2008),
  3425--3441.

\bibitem{Capetanakis}
{\sc Capetanakis, J.}
\newblock Tree algorithms for packet broadcast channels.
\newblock {\em Information Theory, IEEE Transactions on 25}, 5 (sep 1979), 505
  -- 515.

\bibitem{BCC12}
{\sc Censor-Hillel, K., Haeupler, B., Lynch, N., and Médard, M.}
\newblock Bounded-contention coding for wireless networks in the high snr
  regime.
\newblock {\em Distributed Computing 7611\/} (2012), 91--105.

\bibitem{CK}
{\sc Chlamtac, I., and Kutten., S.}
\newblock On broadcasting in radio networks: Problem analysis and protocol
  design.
\newblock {\em IEEE Transactions on Communications 33}, 12 (1985), 1240--1246.

\bibitem{CKR}
{\sc Chlebus, B., Kowalski, D., and Radzik, T.}
\newblock Many-to-many communication in radio networks.
\newblock {\em Algorithmica 54}, 1 (2009), 118--139.

\bibitem{chlebus2012electing}
{\sc Chlebus, B.~S., Kowalski, D.~R., and Pelc, A.}
\newblock Electing a leader in multi-hop radio networks.
\newblock In {\em Principles of Distributed Systems}. Springer, 2012,
  pp.~106--120.

\bibitem{CGL}
{\sc Christersson, M., Gasieniec, L., and Lingas, A.}
\newblock Gossiping with bounded size messages in ad hoc radio networks.
\newblock In {\em The Proceedings of the International Colloquium on Automata,
  Languages and Programming\/} (2002), pp.~377--389.

\bibitem{clark1991unit}
{\sc Clark, B.~N., Colbourn, C.~J., and Johnson, D.~S.}
\newblock Unit disk graphs.
\newblock {\em Annals of Discrete Mathematics 48\/} (1991), 165--177.

\bibitem{CMS01}
{\sc Clementi, A., Monti, A., and Silvestri, R.}
\newblock Selective families, superimposed codes, and broadcasting on unknown
  radio networks.
\newblock In {\em The Proceedings of ACM-SIAM Symposium on Discrete
  Algorithms\/} (2001), pp.~709--718.

\bibitem{CMS03}
{\sc Clementi, A., Monti, A., and Silvestri, R.}
\newblock Distributed broadcast in radio networks of unknown topology.
\newblock {\em Theoretical Computupter Science 302}, 1-3 (June 2003), 337--364.

\bibitem{CK10}
{\sc Cornejo, A., and Kuhn, F.}
\newblock Deploying wireless networks with beeps.
\newblock In {\em Proceedings of the International Conference on Distributed
  Computing (DISC)\/} (2010), pp.~148--162.

\bibitem{CR}
{\sc Czumaj, A., and Rytter, W.}
\newblock Broadcasting algorithms in radio networks with unknown topology.
\newblock In {\em The Proceedings of the Symposium on Foundations of Computer
  Science\/} (2003), pp.~492--501.

\bibitem{Daum:2013}
{\sc Daum, S., Gilbert, S., Kuhn, F., and Newport, C.}
\newblock Broadcast in the ad hoc sinr model.
\newblock In {\em Distributed Computing}. Springer, 2013, pp.~358--372.

\bibitem{SI-Codes-Survey}
{\sc Dyachkov, A.~G., and Rykov, V.~V.}
\newblock A survey of superimposed code theory.
\newblock {\em Problems of Control and Informatzon Theory (English translation)
  v01.12}, 4 (1983).

\bibitem{faloutsos1984signature}
{\sc Faloutsos, C., and Christodoulakis, S.}
\newblock Signature files: An access method for documents and its analytical
  performance evaluation.
\newblock {\em ACM Transactions on Information Systems (TOIS) 2}, 4 (1984),
  267--288.

\bibitem{GPX05}
{\sc Gasieniec, L., Peleg, D., and Xin, Q.}
\newblock Faster communication in known topology radio networks.
\newblock In {\em Proceedings of the twenty-fourth annual ACM symposium on
  Principles of distributed computing\/} (2005), PODC '05, pp.~129--137.

\bibitem{gesbert2003theory}
{\sc Gesbert, D., Shafi, M., Shiu, D.-s., Smith, P.~J., and Naguib, A.}
\newblock From theory to practice: an overview of mimo space-time coded
  wireless systems.
\newblock {\em Selected Areas in Communications, IEEE Journal on 21}, 3 (2003),
  281--302.

\bibitem{SODA-COMM}
{\sc Ghaffari, M., and Haeupler, B.}
\newblock Fast structuring of radio networks for multi-message communications.
\newblock In {\em The Proceedings of the International Symposium on Distributed
  Computing\/} (2013), pp.~492--506.

\bibitem{FOCS}
{\sc Ghaffari, M., Haeupler, B., and Khabbazian, M.}
\newblock {Randomized Broadcast in Radio Networks with Collision Detection}.
\newblock In {\em The Proceedings of the International Symposium on Principles
  of Distributed Computing\/} (2013).

\bibitem{GLS11}
{\sc Ghaffari, M., Lynch, N., and Sastry, S.}
\newblock Leader election using loneliness detection.
\newblock In {\em Proceedings of the 25th international conference on
  Distributed computing\/} (2011), DISC'11, pp.~268--282.

\bibitem{GW85}
{\sc Greenberg, A.~G., and Winograd, S.}
\newblock A lower bound on the time needed in the worst case to resolve
  conflicts deterministically in multiple access channels.
\newblock {\em J. ACM 32}, 3 (1985), 589--596.

\bibitem{NC2011}
{\sc Haeupler, B.}
\newblock Analyzing network coding gossip made easy.
\newblock In {\em Proceedings of the Forty-third Annual ACM Symposium on Theory
  of Computing\/} (2011), STOC '11, pp.~293--302.

\bibitem{halldorsson2014modeling}
{\sc Halld{\'o}rsson, M.~M.}
\newblock Modeling reality algorithmically: The case of wireless communication.
\newblock In {\em Algorithms for Sensor Systems}. Springer, 2014, pp.~1--5.

\bibitem{halldorsson:2012b}
{\sc Halldorsson, M.~M., and Mitra, P.}
\newblock {Wireless Connectivity and Capacity}.
\newblock In {\em The Proceedings of ACM-SIAM Symposium on Discrete
  Algorithms\/} (2012).

\bibitem{Hayes78}
{\sc Hayes, J.}
\newblock An adaptive technique for local distribution.
\newblock {\em Communications, IEEE Transactions on 26}, 8 (aug 1978), 1178 --
  1186.

\bibitem{indyk1997SIcodes}
{\sc Indyk, P.}
\newblock Deterministic superimposed coding with applications to pattern
  matching.
\newblock In {\em The Proceedings of the Symposium on Foundations of Computer
  Science\/} (1997), pp.~127--136.

\bibitem{jurdzinski:2013}
{\sc Jurdzinski, T., Kowalski, D.~R., and Stachowiak, G.}
\newblock Distributed deterministic broadcasting in uniform-power ad hoc
  wireless networks.
\newblock In {\em Proceedings of the 19th International Conference on
  Fundamentals of Computation Theory\/} (2013), FCT'13, pp.~195--209.

\bibitem{JS05}
{\sc Jurdzinski, T., and Stachowiak, G.}
\newblock Probabilistic algorithms for the wake-up problem in single-hop radio
  networks.
\newblock {\em Theoretical Computer Systems 38}, 3 (May 2005), 347--367.

\bibitem{katti2007embracing}
{\sc Katti, S., Gollakota, S., and Katabi, D.}
\newblock Embracing wireless interference: analog network coding.
\newblock In {\em ACM SIGCOMM\/} (2007), pp.~397--408.

\bibitem{Kautz-Singleton}
{\sc Kautz, W., and Singleton, R.}
\newblock Nonrandom binary superimposed codes.
\newblock {\em Information Theory, IEEE Transactions on 10}, 4 (Oct 1964),
  363--377.

\bibitem{kesselheim:2011}
{\sc Kesselheim, T.}
\newblock {A Constant-Factor Approximation for Wireless Capacity Maximization
  with Power Control in the SINR Model}.
\newblock In {\em The Proceedings of ACM-SIAM Symposium on Discrete
  Algorithms\/} (2011).

\bibitem{KK}
{\sc Khabbazian, M., and Kowalski, D.}
\newblock Time-efficient randomized multiple-message broadcast in radio
  networks.
\newblock In {\em The Proceedings of the International Symposium on Principles
  of Distributed Computing\/} (2011), pp.~373--380.

\bibitem{knuth1973vol}
{\sc Knuth, D.}
\newblock Vol. 3: Sorting and searching.
\newblock {\em Addison-Wesley series in computer science\/} (1973).

\bibitem{KP03}
{\sc Kowalski, D., and Pelc, A.}
\newblock Broadcasting in undirected ad hoc radio networks.
\newblock In {\em The Proceedings of the International Symposium on Principles
  of Distributed Computing\/} (2003), pp.~73--82.

\bibitem{KP09}
{\sc Kowalski, D.~R., and Pelc, A.}
\newblock Leader election in ad hoc radio networks: A keen ear helps.
\newblock In {\em Proceedings of the 36th Internatilonal Collogquium on
  Automata, Languages and Programming: Part II\/} (2009), ICALP '09,
  pp.~521--533.

\bibitem{KM93}
{\sc Kushilevitz, E., and Mansour, Y.}
\newblock An ${O(D log(N/D))}$ lower bound for broadcast in radio networks.
\newblock In {\em The Proceedings of the International Symposium on Principles
  of Distributed Computing\/} (1993), pp.~65--74.

\bibitem{KM}
{\sc Kushilevitz, E., and Mansour, Y.}
\newblock An {$\Omega(D \log (N/D))$} lower bound for broadcast in radio
  networks.
\newblock {\em SIAM Journal on Computing 27}, 3 (1998), 702--712.

\bibitem{kyasanur2005capacity}
{\sc Kyasanur, P., and Vaidya, N.~H.}
\newblock Capacity of multi-channel wireless networks: impact of number of
  channels and interfaces.
\newblock In {\em Proceedings International Conference on Mobile Computing and
  Networking\/} (2005), pp.~43--57.

\bibitem{Lynch96}
{\sc Lynch, N.~A.}
\newblock {\em Distributed Algorithms}.
\newblock Morgan Kaufmann Publishers Inc., San Francisco, CA, USA, 1996.

\bibitem{mooers1948application}
{\sc Mooers, C.~N.}
\newblock {\em Application of random codes to the gathering of statistical
  information}.
\newblock PhD thesis, Massachusetts Institute of Technology, 1948.

\bibitem{moscibroda:2006}
{\sc Moscibroda, T., and Wattenhofer, R.}
\newblock {The Complexity of Connectivity in Wireless Networks}.
\newblock In {\em The Proceedings of IEEE INFOCOM\/} (2006).

\bibitem{NakOla02}
{\sc Nakano, K., and Olariu, S.}
\newblock Uniform leader election protocols for radio networks.
\newblock {\em IEEE Trans. Parallel Distrib. Syst. 13}, 5 (May 2002), 516--526.

\bibitem{Peleg:2000}
{\sc Peleg, D.}
\newblock {\em Distributed Computing: A Locality-sensitive Approach}.
\newblock Society for Industrial and Applied Mathematics, Philadelphia, PA,
  USA, 2000.

\bibitem{scheideler:2008}
{\sc Scheideler, C., Richa, A., and Santi, P.}
\newblock {An {O}(log n) Dominating Set Protocol for Wireless Ad-Hoc Networks
  under the Physical Interference Model}.
\newblock In {\em ACM International Symposium on Mobile Ad Hoc Networking and
  Computing\/} (2008).

\bibitem{willard}
{\sc Willard, D.~E.}
\newblock Log-logarithmic selection resolution protocols in a multiple access
  channel.
\newblock {\em SIAM J. Comput. 15}, 2 (May 1986), 468--477.

\end{thebibliography}
%-----------------------------------------------------------------------------------

%\WithEst{\fullOnly{\appendix \input{ParameterEstimation}}}

\end{document}